\documentclass[a4paper,11pt]{article}
\usepackage{jheppub} 
\usepackage{empheq}
\usepackage{bm}
\usepackage{slashed}
\usepackage[T1]{fontenc} 
\usepackage[all]{xy}
\usepackage{rotating}
\usepackage{amsmath,amscd,amsthm,amssymb,amsfonts} 
\usepackage{mathtools,braket}
\usepackage{bbm}
\usepackage{dcolumn}
\usepackage{multirow}
\usepackage{mathrsfs}
\usepackage{tikz,tikz-cd,tikz-3dplot}
\usetikzlibrary{arrows,arrows.meta,braids,calc,decorations.markings,shapes.geometric,fit}
\usetikzlibrary{fadings}
\usetikzlibrary{patterns}
\usetikzlibrary{shadows.blur}
\usepackage{hyperref}
\usepackage{cleveref}
\crefname{figure}{Fig.}{Figs.}
\usepackage{graphicx,adjustbox}

\usetikzlibrary{perspective}

\usepackage[labelformat=simple]{subcaption}

\makeatletter
\newcommand*{\da@rightarrow}{\mathchar
  "0\hexnumber@\symAMSa 4B\!\!}
\newcommand*{\da@leftarrow}{\mathchar"0\hexnumber@\symAMSa 4C }
\newcommand*{\xdashrightarrow}[2][]{%
  \mathrel{%
    \mathpalette{\da@xarrow{#1}{#2}{}\da@rightarrow{\,}{}}{}%
  }%
}
\newcommand{\xdashleftarrow}[2][]{%
  \mathrel{%
    \mathpalette{\da@xarrow{#1}{#2}\da@leftarrow{}{}{\,}}{}%
  }%
}
\newcommand*{\da@xarrow}[7]{%
  \sbox0{$\ifx#7\scriptstyle\scriptscriptstyle\else\scriptstyle\fi#5#1#6\m@th$}%
  \sbox2{$\ifx#7\scriptstyle\scriptscriptstyle\else\scriptstyle\fi#5#2#6\m@th$}%
  \sbox4{$#7\dabar@\m@th$}%
  \dimen@=\wd0 %
  \ifdim\wd2 >\dimen@
    \dimen@=\wd2 %
  \fi
  \count@=2 %
  \def\da@bars{\dabar@\dabar@}%
  \@whiledim\count@\wd4<\dimen@\do{%
    \advance\count@\@ne
    \expandafter\def\expandafter\da@bars\expandafter{%
      \da@bars
      \dabar@ 
    }%
  }%
  \mathrel{#3}%
  \mathrel{%
    \mathop{\da@bars}\limits
    \ifx\\#1\\%
    \else
      _{\copy0}%
    \fi
    \ifx\\#2\\%
    \else
      ^{\copy2}%
    \fi
  }%
  \mathrel{#4}%
}
\makeatother

\tikzset{
  baseline=(current bounding box.center),
  x=28pt, y=28pt, font=\small,
  >=stealth,
  vertex/.style={inner sep=0pt},
  lvertex/.style={vertex, draw, fill= white, opacity=1, minimum size=9pt, font=\footnotesize},
  gnode/.style={lvertex, shape=circle},
  fnode/.style={lvertex, shape=rectangle},
  q->/.style={->, shorten >=0.5pt},
  c->/.style={q->, thick, blue},
  f->/.style={q->, thick, red, densely dashed},
}

\usetikzlibrary{decorations.markings}
\tikzset{
  ->-/.style={decoration={
  markings,
  mark=at position #1 with {\arrow{>}}},postaction={decorate}},
  -<-/.style={decoration={
  markings,
  mark=at position #1 with {\arrow{<}}},postaction={decorate}}
}
\newtheorem{theorem}{Theorem}[section]
\newtheorem{lemma}[theorem]{Lemma}
\newtheorem{proposition}[theorem]{Proposition}

\DeclareMathOperator*{\End}{End}
\DeclareMathOperator{\Tr}{Tr}
\let\Re\relax
\DeclareMathOperator{\Re}{Re}

\newcommand{\Z}{\mathbb{Z}}
\newcommand{\R}{\mathbb{R}}
\newcommand{\C}{\mathbb{C}}

\newcommand{\bfR}{\mathbf{R}}
\newcommand{\bfV}{\mathbf{V}}

\newcommand{\bb}{\mathsf{b}}
\newcommand{\pp}{\mathsf{p}}
\newcommand{\qq}{\mathsf{q}}
\newcommand{\psit}{\tilde{\psi}}
\newcommand{\xb}{\bar{x}}

\newcommand{\TT}{\mathsf{T}}
\newcommand{\TTb}{\overline{\mathsf{T}}}

\def\cube{
  \begin{scope}[scale=1.5]
    \node[gnode] (b) at (0,1,1) {$b$};
    \node[gnode] (c) at (1,0,1) {$c$};
    \node[gnode] (d) at (1,1,0) {$d$};
    \node[gnode] (h) at (1,1,1) {$h$};
    
    \node[gnode] (a) at (0,0,0) {$a$};
    \node[gnode] (e) at (1,0,0) {$e$};
    \node[gnode] (f) at (0,1,0) {$f$};
    \node[gnode] (g) at (0,0,1) {$g$};
  \end{scope}
}

\def\TEcubeLHSone{
  \begin{scope}[scale=1.5]
    \node[gnode] (b) at (0,1,1) {$b_1$};
    \node[gnode] (c) at (1,0,1) {$b_3$};
    \node[gnode] (d) at (1,1,0) {$b_2$};
    \node[gnode] (h) at (1,1,1) {$d$};
    
    \node[gnode] (a) at (0,0,0) {$a_4$};
    \node[gnode] (e) at (1,0,0) {$c_2$};
    \node[gnode] (f) at (0,1,0) {$c_1$};
    \node[gnode] (g) at (0,0,1) {$c_3$};  
  \end{scope}
}

\def\TEcubeLHStwo{
  \begin{scope}[scale=1.5]
    \node[gnode] (b) at (0,1,1) {$c_4$};
    \node[gnode] (c) at (1,0,1) {$d$};
    \node[gnode] (d) at (1,1,0) {$c_6$};
    \node[gnode] (h) at (1,1,1) {$b_4$};
  
    \node[gnode] (a) at (0,0,0) {$c_1$};
    \node[gnode] (e) at (1,0,0) {$b_2$};
    \node[gnode] (f) at (0,1,0) {$a_3$};
    \node[gnode] (g) at (0,0,1) {$b_1$};  
  \end{scope}
}

\def\TEcubeLHSthree{
  \begin{scope}[scale=1.5]
    \node[gnode] (b) at (0,1,1) {$a_2$};
    \node[gnode] (c) at (1,0,1) {$b_3$};
    \node[gnode] (d) at (1,1,0) {$b_4$};
    \node[gnode] (h) at (1,1,1) {$c_5$};
  
    \node[gnode] (a) at (0,0,0) {$b_1$};
    \node[gnode] (e) at (1,0,0) {$d$};
    \node[gnode] (f) at (0,1,0) {$c_4$};
    \node[gnode] (g) at (0,0,1) {$c_3$};  
  \end{scope}
}

\def\TEcubeLHSfour{
  \begin{scope}[scale=1.5]
    \node[gnode] (b) at (0,1,1) {$c_5$};
    \node[gnode] (c) at (1,0,1) {$c_2$};
    \node[gnode] (d) at (1,1,0) {$c_6$};
    \node[gnode] (h) at (1,1,1) {$a_1$};
  
    \node[gnode] (a) at (0,0,0) {$d$};
    \node[gnode] (e) at (1,0,0) {$b_2$};
    \node[gnode] (f) at (0,1,0) {$b_4$};
    \node[gnode] (g) at (0,0,1) {$b_3$};  
  \end{scope}
}

\def\TEcubeRHSone{
  \begin{scope}[scale=1.5]
      \node[gnode] (b) at (0,1,1) {$a_2$};
      \node[gnode] (c) at (1,0,1) {$a_4$};
      \node[gnode] (d) at (1,1,0) {$a_3$};
      \node[gnode] (h) at (1,1,1) {$d$};
    
      \node[gnode] (a) at (0,0,0) {$b_1$};
      \node[gnode] (e) at (1,0,0) {$c_1$};
      \node[gnode] (f) at (0,1,0) {$c_4$};
      \node[gnode] (g) at (0,0,1) {$c_3$};
  \end{scope}
}

\def\TEcubeRHStwo{
  \begin{scope}[scale=1.5]
    \node[gnode] (b) at (0,1,1) {$d$};
    \node[gnode] (c) at (1,0,1) {$c_2$};
    \node[gnode] (d) at (1,1,0) {$c_6$};
    \node[gnode] (h) at (1,1,1) {$a_1$};
  
    \node[gnode] (a) at (0,0,0) {$c_1$};
    \node[gnode] (e) at (1,0,0) {$b_2$};
    \node[gnode] (f) at (0,1,0) {$a_3$};
    \node[gnode] (g) at (0,0,1) {$a_4$};  
  \end{scope}
}

\def\TEcubeRHSthree{
  \begin{scope}[scale=1.5]
    \node[gnode] (b) at (0,1,1) {$a_2$};
    \node[gnode] (c) at (1,0,1) {$b_3$};
    \node[gnode] (d) at (1,1,0) {$a_1$};
    \node[gnode] (h) at (1,1,1) {$c_5$};
  
    \node[gnode] (a) at (0,0,0) {$a_4$};
    \node[gnode] (e) at (1,0,0) {$c_2$};
    \node[gnode] (f) at (0,1,0) {$d$};
    \node[gnode] (g) at (0,0,1) {$c_3$};  
  \end{scope}
}

\def\TEcubeRHSfour{
  \begin{scope}[scale=1.5]
    \node[gnode] (b) at (0,1,1) {$c_4$};
    \node[gnode] (c) at (1,0,1) {$c_5$};
    \node[gnode] (d) at (1,1,0) {$c_6$};
    \node[gnode] (h) at (1,1,1) {$b_4$};
  
    \node[gnode] (a) at (0,0,0) {$d$};
    \node[gnode] (e) at (1,0,0) {$a_1$};
    \node[gnode] (f) at (0,1,0) {$a_3$};
    \node[gnode] (g) at (0,0,1) {$a_2$};  
  \end{scope}
}

\def\Rzero{
  \draw[->] (c) -- (e);
  \draw[->] (e) -- (d);
  \draw[->] (f) -- (d);
  \draw[->] (b) -- (f);
  \draw[->] (g) -- (b);
  \draw[->] (g) -- (c);
  
  \draw[densely dashed, very thick, ->] (b) -- (d);
  \draw[densely dashed, very thick, ->] (c) -- (d);
  \draw[densely dashed, very thick, ->] (c) -- (b);
  
  \draw[densely dashed, ->] (b) -- (h);
  \draw[densely dashed, ->] (c) -- (h);
  \draw[densely dashed, ->] (h) -- (d);
  
  \draw[densely dashed, very thick, ->] (a) -- (h);
  
  \draw[->] (a) -- (b);
  \draw[->] (c) -- (a);
  \draw[->] (a) -- (d);
  \draw[very thick, ->] (a) -- (e);
  \draw[very thick, ->] (a) -- (f);
  \draw[very thick, ->] (g) -- (a);
}

\def\Rdzero{
  \draw[very thick, ->] (c) -- (e);
  \draw[->] (e) -- (d);
  \draw[->] (f) -- (d);
  \draw[->] (b) -- (f);
  
  \draw[densely dashed, ->] (b) -- (h);
  \draw[very thick, densely dashed, ->] (c) -- (h);
  \draw[densely dashed, ->] (h) -- (d);
  
  \draw[densely dashed, very thick, ->] (b) -- (d);
  \draw[densely dashed, ->] (c) -- (d);
  \draw[densely dashed, ->] (c) -- (b);
  \draw[densely dashed, very thick, ->] (c) -- (f);
  
  \draw[->] (a) -- (e);
  \draw[->] (a) -- (f);
  \draw[->] (g) -- (a);
  \draw[very thick, ->] (a) -- (b);
  \draw[->] (c) -- (a);
  \draw[very thick, ->] (a) -- (d);
  
  \draw[->] (g) -- (b);
  \draw[very thick, ->] (g) -- (c);
}

\def\Rddzero{
  \draw[densely dashed, ->] (b) -- (d);
  \draw[densely dashed, very thick, ->] (c) -- (d);
  \draw[densely dashed, ->] (c) -- (b);
  
  \draw[->] (c) -- (e);
  \draw[->] (e) -- (d);
  \draw[->] (f) -- (d);
  \draw[very thick, ->] (b) -- (f);
  
  \draw[very thick, densely dashed, ->] (b) -- (h);
  \draw[densely dashed, ->] (c) -- (h);
  \draw[densely dashed, ->] (h) -- (d);
  
  \draw[densely dashed, very thick, ->] (b) -- (e);
  
  \draw[->] (a) -- (e);
  \draw[->] (a) -- (f);
  \draw[->] (g) -- (a);
  \draw[very thick, ->] (g) -- (b);
  \draw[->] (g) -- (c);
  
  \draw[->] (b) -- (a);
  \draw[very thick, ->] (c) -- (a);
  \draw[very thick, ->] (a) -- (d);
}

\def\Rdddzero{
  \draw[->] (c) -- (e);
  \draw[->] (e) -- (d);
  \draw[->] (f) -- (d);
  \draw[->] (b) -- (f);
  
  \draw[densely dashed, ->] (b) -- (h);
  \draw[densely dashed, ->] (c) -- (h);
  \draw[densely dashed, ->] (h) -- (d);
  
  \draw[densely dashed, very thick, ->] (b) -- (d);
  \draw[densely dashed, very thick, ->] (c) -- (d);
  \draw[densely dashed, very thick, ->] (c) -- (b);
  \draw[densely dashed, very thick, ->] (a) -- (h);
  
  \draw[very thick, ->] (a) -- (e);
  \draw[very thick, ->] (a) -- (f);
  \draw[very thick, ->] (g) -- (a);
  \draw[->] (g) -- (b);
  \draw[->] (g) -- (c);
  \draw[->] (b) -- (a);
  \draw[->] (c) -- (a);
  \draw[->] (a) -- (d);
}

\def\Rone{
  \draw[->] (c) -- (e);
  \draw[->] (e) -- (d);
  \draw[->] (f) -- (d);
  \draw[->] (b) -- (f);
  \draw[->] (g) -- (b);
  \draw[->] (g) -- (c);
  
  \draw[very thick, densely dashed, ->] (b) -- (h);
  \draw[very thick, densely dashed, ->] (c) -- (h);
  \draw[very thick, densely dashed, ->] (h) -- (d);
  
  \draw[very thick, ->] (e) -- (f);
  \draw[very thick, ->] (g) -- (f);
  \draw[very thick, ->] (g) -- (e);
  \draw[densely dashed, ->] (e) -- (h);
  \draw[densely dashed, ->] (h) -- (f);
  \draw[densely dashed, ->] (g) -- (h);
  \draw[densely dashed, very thick, ->] (a) -- (h);
  
  \draw[->] (a) -- (e);
  \draw[->] (a) -- (f);
  \draw[->] (g) -- (a);
}

\def\Rdone{    
  \draw[densely dashed, very thick, ->] (e) -- (h);
  \draw[densely dashed, ->] (h) -- (f);
  \draw[densely dashed, very thick, ->] (g) -- (h);
  
  \draw[->] (c) -- (e);
  \draw[->] (e) -- (d);
  \draw[very thick, ->] (f) -- (d);
  \draw[very thick, ->] (b) -- (f);
  
  \draw[densely dashed, ->] (b) -- (h);
  \draw[densely dashed, ->] (c) -- (h);
  \draw[densely dashed, ->] (h) -- (d);
  
  \draw[very thick, densely dashed, ->] (c) -- (f);
  \draw[->] (e) -- (f);
  \draw[->] (g) -- (f);
    
  \draw[->] (a) -- (e);
  \draw[very thick, ->] (a) -- (f);
  \draw[->] (g) -- (a);
  
  \draw[->] (g) -- (b);
  \draw[->] (g) -- (c);
  \draw[very thick, ->] (g) -- (e);
}

\def\Rddone{
  \draw[very thick, ->] (c) -- (e);
  \draw[very thick, ->] (e) -- (d);
  \draw[->] (f) -- (d);
  \draw[->] (b) -- (f);
  
  \draw[densely dashed, ->] (b) -- (h);
  \draw[densely dashed, ->] (c) -- (h);
  \draw[densely dashed, ->] (h) -- (d);
  
  \draw[->] (e) -- (f);
  \draw[densely dashed, ->] (h) -- (e);
  \draw[very thick, densely dashed, ->] (h) -- (f);
  \draw[densely dashed, very thick, ->] (b) -- (e);
  
  \draw[very thick, ->] (a) -- (e);
  \draw[->] (a) -- (f);
  \draw[->] (g) -- (a);
  \draw[->] (g) -- (b);
  \draw[->] (g) -- (c);
  \draw[very thick, ->] (g) -- (f);
  \draw[->] (g) -- (e);
  \draw[very thick, densely dashed, ->] (g) -- (h);
}

\def\Rdddone{
  \draw[densely dashed, ->] (h) -- (e);
  \draw[densely dashed, ->] (h) -- (f);
  \draw[densely dashed, ->] (g) -- (h);
  
  \draw[->] (c) -- (e);
  \draw[->] (e) -- (d);
  \draw[->] (f) -- (d);
  \draw[->] (b) -- (f);
  
  \draw[very thick, densely dashed, ->] (b) -- (h);
  \draw[very thick, densely dashed, ->] (c) -- (h);
  \draw[very thick, densely dashed, ->] (h) -- (d);
  
  \draw[densely dashed, very thick, ->] (a) -- (h);
  
  \draw[->] (g) -- (b);
  \draw[->] (g) -- (c);
  
  \draw[very thick, ->] (e) -- (f);
  \draw[very thick, ->] (g) -- (f);
  \draw[very thick, ->] (g) -- (e);
  \draw[->] (a) -- (e);
  \draw[->] (a) -- (f);
  \draw[->] (g) -- (a);
}

\title{Solving the tetrahedron equation by Teichm\"uller TQFT}

\author[1]{Myungbo Shim,}
\author[2,1]{Xiaoyue Sun,}
\author[1]{Hao Ellery Wang,}
\author[1]{Junya Yagi}

\affiliation[1]{Yau Mathematical Sciences Center, Tsinghua University\\Beijing 100084, China}
\affiliation[2]{Beijing Institute of Mathematical Sciences and Applications
\\Beijing 101408, China}

\emailAdd{mbshim@tsinghua.edu.cn}
\emailAdd{sunxiaoyue@bimsa.cn}
\emailAdd{wh2022@tsinghua.edu.cn}
\emailAdd{junyagi@tsinghua.edu.cn}

\abstract{We propose an approach to construct three-dimensional lattice models
  using line defects in state integral models on shaped triangulations
  of $3$-manifolds.  The Boltzmann weights for these models satisfy a
  variant of the tetrahedron equation, which implies integrability
  under suitable assumptions on R-matrices and transfer matrices.  As
  an explicit example, we present a solution produced by Teichm\"uller
  TQFT.}

\begin{document}

\maketitle

\section{Introduction}

The tetrahedron equation, introduced by Zamolodchikov
\cite{MR611994b, Zamolodchikov:1981kf}, is a three-dimensional analog
of the Yang--Baxter equation and plays a fundamental role in
three-dimensional integrable lattice models and $(2+1)$-dimensional
integrable quantum field theories. Compared with the Yang--Baxter
equation, however, the tetrahedron equation is substantially more
complicated, and our understanding of it remains somewhat limited.  For
recent developments, see e.g.\ \cite{Gavrylenko:2020eov,
  Yagi:2022tot, MR4616664, Sun:2022mpy, Inoue:2023vtx, Inoue:2023rer,
  Inoue:2024swb, Padmanabhan:2024zma, Yagi:2024tac, MR4850776,
  Inoue:2025xzj, BKMS, IK}.

In this work, we present a geometric construction of solutions to a
variant of the tetrahedron equation, consisting of a pair of equations
which we call the bicolored tetrahedron equations (BTEs), using state
integral models on shaped triangulations of $3$-manifolds.%
\footnote{The BTEs have appeared in the literature under the name
  ``modified tetrahedron equation''; see
  e.g.~\cite{Mangazeev:1993fu, Mangazeev:1993qw}. We prefer the term “bicolored,” which
  more transparently reflects the structure of the modification.}
Solutions of the BTEs can be used to construct three-dimensional
bipartite lattice models, which are integrable if the solutions have
suitable properties.

Our construction is based on the following three observations:
\begin{itemize}
\item For interaction-round-a-cube (IRC) models \cite{MR696804}, the
  tetrahedron equation admits a graphical interpretation as the
  equivalence of two decompositions of the rhombic dodecahedron into
  four cubes, with each cube corresponding to an R-matrix, as
  illustrated in Figure~\ref{fig:RD}.

\item In topological quantum field theories (TQFTs) of Turaev--Viro
  type \cite{Turaev:1992hq}, the partition function is defined using a
  triangulation of the spacetime $3$-manifold but is independent of the
  choice of triangulation.

\item In state integral models such as Teichm\"uller TQFT
  \cite{MR3227503, AK2}, which are defined on triangulated
  $3$-manifolds equipped with an additional structure called a shape
  structure, one can introduce line defects around which the total
  angle is not equal to $2\pi$.
\end{itemize}
Geometrically, what we do is to consider shaped triangulations of the
rhombic dodecahedron representing each side of each BTE, and
demonstrate that the two sides of each BTE are connected by a sequence
of shaped $2$--$3$ moves.

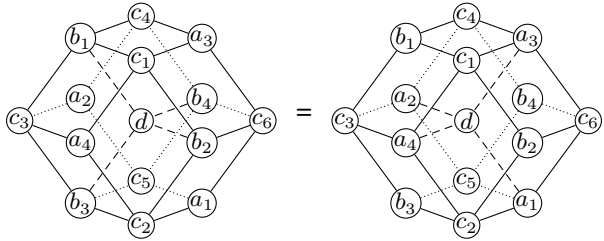
\begin{figure}[t]
  \begin{equation}
    \tdplotsetmaincoords{60}{270}
    \begin{tikzpicture}[tdplot_main_coords, scale=0.8]
    \node[gnode] (d) at (0,0,0) {$d$};

    \node[gnode] (a1) at (-1,-1,-1) {$a_1$};
    \node[gnode] (a2) at (-1,1,1) {$a_2$};
    \node[gnode] (a3) at (1,-1,1) {$a_3$};
    \node[gnode] (a4) at (1,1,-1) {$a_4$};
    
    \node[gnode] (b1) at (1,1,1) {$b_1$};
    \node[gnode] (b2) at (1,-1,-1) {$b_2$};
    \node[gnode] (b3) at (-1,1,-1) {$b_3$};
    \node[gnode] (b4) at (-1,-1,1) {$b_4$};

    \node[gnode] (c1) at (2,0,0) {$c_1$};
    \node[gnode] (c2) at (0,0,-2) {$c_2$};
    \node[gnode] (c3) at (0,2,0) {$c_3$};
    \node[gnode] (c4) at (0,0,2) {$c_4$};
    \node[gnode] (c5) at (-2,0,0) {$c_5$};
    \node[gnode] (c6) at (0,-2,0) {$c_6$};
    
    \draw[densely dotted] (b3) -- (c5);
    \draw (b3) -- (c3);
    \draw (b3) -- (c2);
    \draw (b1) -- (c4);
    \draw(b1) -- (c1);
    \draw(b1) -- (c3);
    \draw(b2) -- (c6);
    \draw(b2) -- (c1);
    \draw(b2) -- (c2);
    \draw[densely dotted] (b4) -- (c6);
    \draw[densely dotted] (b4) -- (c5);
    \draw[densely dotted] (b4) -- (c4);
    \draw (a3) -- (c6);
    \draw (a3) -- (c4);
    \draw (a3) -- (c1);
    \draw(a1) -- (c6);
    \draw[densely dotted] (a1) -- (c5);
    \draw(a1) -- (c2);
    \draw[densely dotted] (a2) -- (c5);
    \draw[densely dotted] (a2) -- (c4);
    \draw[densely dotted] (a2) -- (c3);
    \draw (a4) -- (c1);
    \draw (a4) -- (c3);
    \draw (a4) -- (c2);
    \draw[densely dashed] (b3) -- (d);
    \draw[densely dashed] (b1) -- (d);
    \draw[densely dashed] (b2) -- (d);
    \draw[densely dashed] (b4) -- (d);
  \end{tikzpicture}
  \ = \
  \begin{tikzpicture}[tdplot_main_coords, scale=0.8]
    \node[gnode] (d) at (0,0,0) {$d$};

    \node[gnode] (a1) at (-1,-1,-1) {$a_1$};
    \node[gnode] (a2) at (-1,1,1) {$a_2$};
    \node[gnode] (a3) at (1,-1,1) {$a_3$};
    \node[gnode] (a4) at (1,1,-1) {$a_4$};
    
    \node[gnode] (b1) at (1,1,1) {$b_1$};
    \node[gnode] (b2) at (1,-1,-1) {$b_2$};
    \node[gnode] (b3) at (-1,1,-1) {$b_3$};
    \node[gnode] (b4) at (-1,-1,1) {$b_4$};

    \node[gnode] (c1) at (2,0,0) {$c_1$};
    \node[gnode] (c2) at (0,0,-2) {$c_2$};
    \node[gnode] (c3) at (0,2,0) {$c_3$};
    \node[gnode] (c4) at (0,0,2) {$c_4$};
    \node[gnode] (c5) at (-2,0,0) {$c_5$};
    \node[gnode] (c6) at (0,-2,0) {$c_6$};
    
    \draw[densely dotted] (b3) -- (c5);
    \draw(b3) -- (c3);
    \draw(b3) -- (c2);
    \draw(b1) -- (c4);
    \draw(b1) -- (c1);
    \draw(b1) -- (c3);
    \draw(b2) -- (c6);
    \draw(b2) -- (c1);
    \draw(b2) -- (c2);
    \draw[densely dotted] (b4) -- (c6);
    \draw[densely dotted] (b4) -- (c5);
    \draw[densely dotted] (b4) -- (c4);
    \draw(a3) -- (c6);
    \draw(a3) -- (c4);
    \draw(a3) -- (c1);
    \draw(a1) -- (c6);
    \draw[densely dotted] (a1) -- (c5);
    \draw(a1) -- (c2);
    \draw[densely dotted] (a2) -- (c5);
    \draw[densely dotted] (a2) -- (c4);
    \draw[densely dotted] (a2) -- (c3);
    \draw(a4) -- (c1);
    \draw(a4) -- (c3);
    \draw(a4) -- (c2);
    \draw[densely dashed] (a3) -- (d);
    \draw[densely dashed] (a1) -- (d);
    \draw[densely dashed] (a2) -- (d);
    \draw[densely dashed] (a4) -- (d);
  \end{tikzpicture}\nonumber
\end{equation}
\caption{The tetrahedron equation for IRC models.}
\label{fig:RD}
\end{figure}

A similar idea was proposed by Maillet \cite{Maillet:1993av}, who
related the tetrahedron equation to the pentagon identity for an
operator represented by a tetrahedron.  From this viewpoint, the
tetrahedron equation arises as the condition for trivial holonomy of
the parallel transport operator from three adjacent faces of a cube
consisting of six tetrahedra to the opposite three faces.  However,
when this idea is implemented using Turaev--Viro TQFTs, the resulting
partition function is a topological invariant and therefore does not
encode the size of the lattice.%
\footnote{Besides Maillet's work, there have been several solutions to
  the tetrahedron equation and its variants based on the pentagon
  identity of quantum dilogarithms. See, e.g., \cite{MR1338071,
    Maillard:1997ypg, BKMS}.}

The crucial difference between our construction and Maillet's
construction applied to Turaev--Viro TQFTs is that when our R-matrix
is used to build a three-dimensional cubic lattice, the corresponding
triangulated manifold necessarily contains a network of line defects.
Consequently, the associated partition function is not a topological
invariant and depends on the size of the lattice.

The organization of this paper is as follows.  

In section~\ref{sec:BLM}, we discuss three-dimensional bipartite
lattice models and their integrability.  We explain how the BTEs
imply integrability if the transfer matrix satisfies a certain
nondegeneracy condition (Lemma~\ref{lemma:one-dimensional
  eigenspace2}).

In section~\ref{sec:SIM-BTE}, we present the geometric construction of
solutions to the BTEs.  For our choice of triangulations, it is
necessary to introduce defects in the cubes and the rhombic dodecahedra.
We identify a choice of shape structures (Figure~\ref{fig:cubes} and
Figure~\ref{fig:parameters}) that solves the BTEs
(Proposition~\ref{prop:BTEs-cubes}).

In section~\ref{sec:Teichmuller}, we prove that the R-matrices
produced by Teichm\"uller TQFT exactly solve the BTEs.

\section{Three-dimensional bipartite lattice models}
\label{sec:BLM}

In this section, we discuss classical spin models on bipartite cubic
lattices and their integrability.

\subsection{Spin models on bipartite cubic lattices}

A \emph{bipartite graph} is a graph whose vertex is colored either
black or white, $[0]$ or $[1] \in \Z/2\Z$.  Each edge connects exactly
one black vertex and one white vertex.  In what follows, we will
consider classical spin models on a periodic bipartite cubic lattice
of size $2L \times 2M \times 2N$, where the two vertex colors
represent two kinds of interactions.

The lattice consists of three sets of planes intersecting in a
$3$-torus parametrized by coordinates $(x^1, x^2, x^3)$ with periodic
identification
$(x^1, x^2, x^3) \sim (x^1 + 2L, x^2, x^3) \sim (x^1, x^2 + 2M, x^3)
\sim (x^1, x^2, x^3 + 2N)$. We take these planes to be
\begin{itemize}
\item planes $1_l$ located at $x^1 = l$, $l = 1$, \dots, $2L$;

\item planes $2_m$ located at $x^2 = m$, $m = 1$, \dots, $2M$; and

\item planes $3_n$ located at $x^3 = n$, $n = 1$, \dots, $2N$.
\end{itemize}
We order the planes lexicographically:
\begin{equation*}
  1_1 < 1_2 < \dotsb < 1_{2L}
  < 2_1 < 2_2 < \dotsb < 2_{2M}
  < 3_1 < 3_2 < \dotsb < 3_{2N} \,.
\end{equation*}
The intersections of planes $\alpha$ and $\beta$ with $\alpha < \beta$
are lines $(\alpha\beta)$.  The intersections of planes $\alpha$,
$\beta$ and $\gamma$ with $\alpha < \beta < \gamma$ are vertices
$(\alpha\beta\gamma)$.  The lines are oriented toward vertices of
larger lexicographic order.  We assign color $[l + m + n]$ to vertex
$(1_l 2_m 3_n)$.  See Figure~\ref{fig:lattice}.

\begin{figure}
  \centering
  \begin{subfigure}{0.25\textwidth}
    \tdplotsetmaincoords{-75}{-35}
    \begin{tikzpicture}[tdplot_main_coords, scale=2, font=\scriptsize]
      \draw[fill=blue, opacity=0.2]
      (-0.1,0.5,-0.1) -- (-0.1,0.5,1.1) node[xshift=-2pt, yshift=-2pt, opacity=1] {$3_n$}
      -- (1.1,0.5,1.1) -- (1.1,0.5,-0.1) -- cycle;
      
      \draw[fill=red, opacity=0.2]
      (0.5,-0.1,-0.1) -- (0.5,1.1,-0.1) node[xshift=-4pt, yshift=-2pt, opacity=1] {$1_l$}
      -- (0.5,1.1,1.1) -- (0.5,-0.1,1.1) -- cycle;
      
      \draw[fill=green, opacity=0.2]
      (-0.1,-0.1,0.5) -- (1.1,-0.1,0.5) node[xshift=-2pt, yshift=6pt, opacity=1] {$2_m$}
      -- (1.1,1.1,0.5) -- (-0.1,1.1,0.5) -- cycle;
      
      \draw[->, thick] (-0.2,0.5,0.5) node[left, xshift=2pt, yshift=-4pt] {$(2_m 3_n)$} -- (1.2,0.5,0.5);
      \draw[->, thick] (0.5,-0.2,0.5) node[below right, xshift=-4pt, yshift=2pt] {$(1_l 2_m)$} -- (0.5,1.2,0.5);
      \draw[->, thick] (0.5,0.5,-0.2) node[above, yshift=-2pt] {$(1_l 3_n)$} -- (0.5,0.5,1.2);

      \node[draw, thick, fill=black, shape=circle, inner sep=2pt] at (0.5,0.5,0.5) {};

    \end{tikzpicture}
    
    \caption{}
    \label{fig:vertex}
  \end{subfigure}
  \hfill
  \begin{subfigure}{0.65\textwidth}
  \centering
  \begin{tikzpicture}[3d view={40}{25}, scale=1, font=\scriptsize]
    \begin{scope}[shift={(0.5,0,-1.3)}, scale=0.5]
      \draw [->] (0,0,0) -- (1,0,0) node[xshift=5pt, yshift=1pt] {$x^1$};
      \draw [->] (0,0,0) -- (0,1,0) node[xshift=5pt, yshift=3pt] {$x^2$};
      \draw [->] (0,0,0) -- (0,0,1) node[xshift=2pt, yshift=4pt] {$x^3$};
    \end{scope}

    \foreach \l in {1,2,3,5} {
      \draw[->, thick] (\l,0.5,0)
      node[xshift=-8pt, yshift=-6pt] {$(1_{\ifnum\l=5 2L\else\inteval\l\fi} 3_n)$}
      -- (\l,5.5,0);
      \draw[very thick, white] (\l,3.5,0) -- (\l,4.5,0);
      \draw[very thick, dotted] (\l,3.6,0) -- (\l,4.4,0);
    }
    
    \foreach \m in {1,2,3,5} {
      \draw[->, thick] (0.5,\m,0)
      node[xshift=-10pt, yshift=6pt] {$(2_{\ifnum\m=5 2M\else\inteval{\m}\fi} 3_n)$}
      -- (5.5,\m,0);
      \draw[very thick, white] (3.5,\m,0) -- (4.5,\m,0);
      \draw[very thick, dotted] (3.6,\m,0) -- (4.4,\m,0);
    }

    \foreach \l in {2,3,5} {
      \foreach \m [evaluate={\inteval{\l+\m}} as \lpm] in {1,2,3} {
        \draw[->, thick] (\l,\m,-0.5) -- (\l,\m,0.5);
        \ifodd\lpm
          \node[gnode, minimum size=5pt, fill=black] at (\l,\m,0) {};
        \else
          \node[gnode, minimum size=5pt, fill=white] at (\l,\m,0) {};
        \fi
      }
    }
    \foreach \l in {2,3} {
        \draw[->, thick] (\l,5,-0.5) -- (\l,5,0.5);
        \ifodd\l
          \node[gnode, minimum size=5pt, fill=black] at (\l,5,0) {};
        \else
          \node[gnode, minimum size=5pt, fill=white] at (\l,5,0) {};
        \fi
    }
    
        \foreach \l in {5} {
      \foreach \m [evaluate={\inteval{\l+\m}} as \lpm] in {1,2,3} {
        \draw[->, thick] (\l,\m,-0.5) -- (\l,\m,0.5);
        \ifodd\lpm
          \node[gnode, minimum size=5pt, fill=white] at (\l,\m,0) {};
        \else
          \node[gnode, minimum size=5pt, fill=black] at (\l,\m,0) {};
        \fi
      }
    }
 \foreach \l in {5} {
      \foreach \m [evaluate={\inteval{\l+\m}} as \lpm] in {5} {
        \draw[->, thick] (\l,\m,-0.5) -- (\l,\m,0.5);
        \ifodd\lpm
          \node[gnode, minimum size=5pt, fill=black] at (\l,\m,0) {};
        \else
          \node[gnode, minimum size=5pt, fill=white] at (\l,\m,0) {};
        \fi
      }
    }

    \foreach \m in {1,2,3} {
      \draw[->, thick] (1,\m,-0.5) -- (1,\m,0.5);
      \ifodd\m
        \node[gnode, minimum size=5pt, fill=white] at (1,\m,0) {};
      \else
        \node[gnode, minimum size=5pt, fill=black] at (1,\m,0) {};
      \fi
    }
    
    \draw[->, thick] (1,5,-0.5) -- (1,5,0.5) node[gnode, minimum size=5pt, fill=black] at (1,5,0) {};
  \end{tikzpicture}
  
  \caption{}
  \label{fig:T}
\end{subfigure}

\caption{(a) A neighborhood of vertex $(1_l 2_m 3_n)$ with $l+m+n$
  even.  (b) A neighborhood of plane $3_n$ with $n$ odd.}
  \label{fig:lattice}
\end{figure}
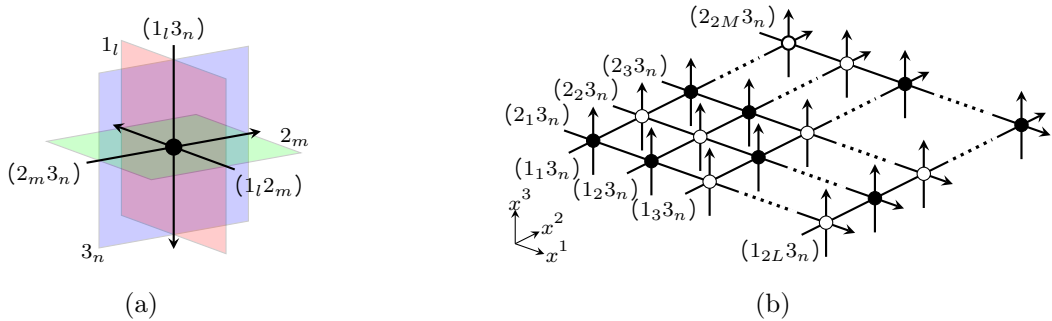

Spin models are statistical mechanics models of interacting state
variables (or ``spins'') located in space according to some patterns.
Different types of spin models can be constructed on the bipartite
cubic lattice depending on where spins are located and how they
interact.  We will consider models in which classical spins interact
locally at the vertices of the lattice.  Spins may be located on
\begin{itemize}
\item the \emph{edges} formed by the intersection of the planes;

\item the \emph{faces} on the planes bounded by the edges; and

\item the \emph{regions} in the $3$-torus bounded by the planes.
\end{itemize}
A single vertex is surrounded by 26 spins (6 on edges, 12 on faces,
and 8 in regions).  The local Boltzmann weight $W$ describing their
interaction is a function of the vertex color and the values of the 26
spins, as well as the \emph{spectral parameters} assigned to the three
planes meeting at the vertex:
\begin{equation*}
  W\colon \Z_2 \times C^3
  \times R_e^6 \times  R_f^{12} \times  R_r^8
  \to \C \,.
\end{equation*}
Here, $C$ is the set of possible values of spectral parameters, and
$R_e$, $R_f$ and $R_r$ are the sets in which spins on edges, faces,
and regions take values, respectively.  The Boltzmann weight of the
model is the product of local Boltzmann weights over all vertices.

A model that has spins only on the edges is called a \emph{vertex
  model}.  Suppose that in a vertex model the six spins adjacent to
vertex $(\alpha\beta\gamma)$ of color $[\sigma]$ take values $i$, $j$,
$k$, $i'$, $j'$, $k'$, as shown in Figure~\ref{fig:spins-vertex}.
Following the standard convention, we will write
\begin{equation}
  \label{eq:R}
  R^{[\sigma]}(r_\alpha, r_\beta, r_\gamma)_{ijk}^{i'j'k'}
\end{equation}
to denote the local Boltzmann weight for this configuration, where
$r_\alpha$, $r_\beta$, $r_\gamma$ are the spectral parameters of
planes $\alpha$, $\beta$, $\gamma$.

A model on a cubic lattice with spins living only on the regions is
traditionally called an \emph{Interaction-Round-a-Cube (IRC) model}.
The eight spins at the corners of a cube interact around the cube.
The local Boltzmann weight around vertex $(\alpha\beta\gamma)$ in an
IRC model is denoted by
\begin{equation*}
  W^{[\sigma]}(a|efg|bcd|h; r_\alpha, r_\beta, r_\gamma) \,,
\end{equation*}
where the eight spins with values $a$, $b$, $c$, $d$, $e$, $f$, $g$,
$h$ are arranged as in Figure~\ref{fig:spins-IRC}.

  


\begin{figure}
  \centering
  \begin{subfigure}{0.3\textwidth}
    \tdplotsetmaincoords{-75}{-35}
    \begin{tikzpicture}[tdplot_main_coords, scale=1.8, font=\scriptsize]
      \draw[->, thick] (-0.2,0.5,0.5) -- (1.2,0.5,0.5);
      \draw[->, thick] (0.5,-0.2,0.5) -- (0.5,1.2,0.5);
      \draw[->, thick] (0.5,0.5,-0.2) -- (0.5,0.5,1.2);

      \draw[->, thick] (-0.2,0.5,0.5) node[left, xshift=2pt, yshift=-4pt] {$(\beta \gamma)$} -- (1.2,0.5,0.5);
      \draw[->, thick] (0.5,-0.2,0.5) node[below right, xshift=-4pt, yshift=2pt] {$(\alpha \beta)$} -- (0.5,1.2,0.5);
      \draw[->, thick] (0.5,0.5,-0.2) node[above, yshift=-2pt] {$(\alpha \gamma)$} -- (0.5,0.5,1.2);

      \node[draw, thick, fill=black, shape=circle, inner sep=1.5pt] at (0.5,0.5,0.5) {};

      \node[gnode] at (0.5,0.1,0.5) {$i$};
      \node[gnode] at (0.5,0.9,0.5) {$i'$};

      \node[gnode] at (0.5,0.5,0.1) {$j$};
      \node[gnode] at (0.5,0.5,0.9) {$j'$};

      \node[gnode] at (0.1,0.5,0.5) {$k$};
      \node[gnode] at (0.9,0.5,0.5) {$k'$};
    \end{tikzpicture}
    
    \caption{}
    \label{fig:spins-vertex}
  \end{subfigure}
  \qquad
  \begin{subfigure}{0.2\textwidth}
    \tdplotsetmaincoords{-75}{-35}
    \begin{tikzpicture}[tdplot_main_coords, scale=1.5]
      \node[gnode] (b) at (0,1,1) {$b$};
      \node[gnode] (c) at (1,0,1) {$c$};
      \node[gnode] (d) at (1,1,0) {$d$};
      \node[gnode] (h) at (1,1,1) {$h$};
      
      \draw[densely dotted] (b) -- (h);
      \draw[densely dotted] (c) -- (h);
      \draw[densely dotted] (d) -- (h);

      \draw[fill=blue, opacity=0.2]
      (-0.1,0.5,-0.1) -- (-0.1,0.5,1.1) node[xshift=-2pt, yshift=-2pt, opacity=1] {$\gamma$}
      -- (1.1,0.5,1.1) -- (1.1,0.5,-0.1) -- cycle;
      
      \draw[fill=red, opacity=0.2]
      (0.5,-0.1,-0.1) -- (0.5,1.1,-0.1) node[xshift=-4pt, yshift=-2pt, opacity=1] {$\alpha$}
      -- (0.5,1.1,1.1) -- (0.5,-0.1,1.1) -- cycle;
      
      \draw[fill=green, opacity=0.2]
      (-0.1,-0.1,0.5) -- (1.1,-0.1,0.5) node[xshift=-2pt, yshift=6pt, opacity=1] {$\beta$}
      -- (1.1,1.1,0.5) -- (-0.1,1.1,0.5) -- cycle;
      
      \draw[->, thick] (-0.2,0.5,0.5) -- (1.2,0.5,0.5);
      \draw[->, thick] (0.5,-0.2,0.5) -- (0.5,1.2,0.5);
      \draw[->, thick] (0.5,0.5,-0.2) -- (0.5,0.5,1.2);

      \node[draw, thick, fill=black, shape=circle, inner sep=1.5pt] at (0.5,0.5,0.5) {};

      \node[gnode] (a) at (0,0,0) {$a$};
      \node[gnode] (e) at (1,0,0) {$e$};
      \node[gnode] (f) at (0,1,0) {$f$};
      \node[gnode] (g) at (0,0,1) {$g$};
      
      \draw[black!50] (a) -- (e);
      \draw[black!50] (a) -- (f);
      \draw[black!50] (a) -- (g);
      \draw[black!50] (e) -- (c);
      \draw[black!50] (e) -- (d);
      \draw[black!50] (f) -- (d);
      \draw[black!50] (f) -- (b);
      \draw[black!50] (g) -- (b);
      \draw[black!50] (g) -- (c);
    \end{tikzpicture}
    
    \caption{}
    \label{fig:spins-IRC}
  \end{subfigure}

\caption{Spins around vertex $(\alpha\beta\gamma)$ in (a) a vertex
  model and (b) an IRC model.}
  \label{fig:spins}
\end{figure}
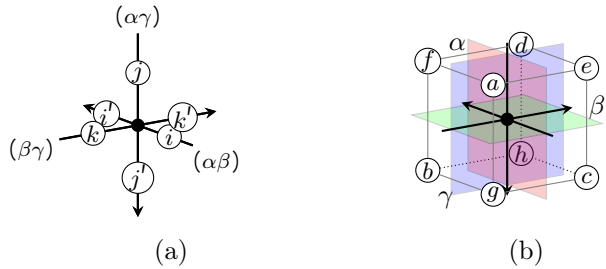

Given a model of the type described above, we can always reformulate
it as a vertex model.  This can be done in two steps.

First, we replace the spin in each region with six spins and
distribute them to the surrounding six faces.  Now spins are only on
the edges and the faces.  Each face has three spins (the original spin
plus the two coming from the adjacent regions), and we combine them
into a single spin valued in $\widetilde{S}_f = R_f \times R_r^2$.
Using this new set of spins, we define a model whose partition
function is equal to the partition function of the original model: if
a spin configuration can arise from that in the original model, we
assign to it the corresponding original Boltzmann weight; otherwise,
we set the Boltzmann weight to zero.

Next, we convert the model just obtained into a vertex model.  We
replace the spin on each face with four spins and distribute them to
the four surrounding edges.  Then, all spins are located on edges, and
they take values in $\widetilde{S}_e = R_e \times \widetilde{S}_f^4$.
Using these spins, we can define a vertex model equivalent to the
model constructed in the previous paragraph, hence to the original
model, by choosing the local Boltzmann weight appropriately.

\subsection{Bicolored tetrahedron equations and integrability}

As explained above, any model of the type considered in the present
paper can be reformulated as a vertex model on the bipartite cubic
lattice.  Introducing the vector space
\begin{equation}
  V = \bigoplus_i \C \ket{i} \,,
\end{equation}
define an operator-valued function
$R\colon \Z_2 \times C^3 \to \End(V \otimes V \otimes V)$ by
\begin{equation}
  R^{[\sigma]}(r_\alpha, r_\beta, r_\gamma) \ket{i} \otimes \ket{j} \otimes \ket{k}
  =
  \sum_{i',j',k'} R^{[\sigma]}(r_\alpha, r_\beta, r_\gamma)_{ijk}^{i'j'k'}
  \ket{i'} \otimes \ket{j'} \otimes \ket{k'}
  \,.
\end{equation}
In the following we discuss conditions on $R$ that imply the
integrability of this vertex model.

Suppose that the model is defined on a periodic
$2L \times 2M \times 2N$ bipartite cubic lattice, which we think of as
consisting of $2N$ layers of periodic $2L \times 2M$ bipartite square
lattices that are parallel to the $x^1$-$x^2$ plane and stacked in the $x^3$
direction.  Take one of those layers.  Its neighborhood is illustrated
in Figure~\ref{fig:T}.  Let $V_{(\alpha\beta)}$ denote the copy of
vector space $V$ assigned to the line $(\alpha\beta)$, and let
$R^{[\sigma]}_{(\alpha\beta\gamma)}$ be the operator acting on
$V_{(\alpha\beta)} \otimes V_{(\alpha\gamma)} \otimes
V_{(\beta\gamma)}$ by $R^{[\sigma]}(r_\alpha, r_\beta, r_\gamma)$.

    



Consider the intersections of the planes $1_1$, \dots, $1_{2L}$ and
two planes $\alpha$, $\beta$.  The Boltzmann weight for spin
configurations in their neighborhood is given by
\begin{equation}
  \bfR^{[\sigma]}_{\alpha\beta}
  :=
  \Tr_{V_{(\alpha\beta)}}\biggl(\prod_l R^{[\sigma+l]}_{(1_l\alpha\beta)}\biggr)
  \in
  \End(\bfV_\alpha \otimes \bfV_\beta)
  \,,
\end{equation}
where
\begin{equation}
  \bfV_\alpha := \bigotimes_l V_{(1_l\alpha)} \,.
\end{equation}
This defines an operator
$\bfR\colon \Z_2 \times C^{2L} \times C^2 \to \End(\bfV \otimes \bfV)$
with $\bfV = V^{\otimes 2L}$, called the \emph{trace reduction} of
$R$.
%

Similarly, the neighborhood of plane $3_n$ defines the \emph{layer
  transfer matrix}
$T \colon \Z_2 \times C \to \End(\bigotimes_{l,m} V_{(1_l 2_m)})$ by
\begin{equation}
  \label{eq:T}
  T^{[\sigma]}(r_{3_n})
  :=
  \Tr_{\bigotimes_{l,m} V_{(1_l 3_n)} \otimes V_{(2_m 3_n)}}\biggl(
  \prod_{m,n} R^{[\sigma+l+m]}_{(1_l 2_m 3_n)}\biggr)
  =
  \Tr_{\bfV_{3_n}}\biggl(\prod_m \bfR^{[\sigma+m]}_{2_m 3_n}\biggr)
  \,,
\end{equation}
where the operator product is taken in the lexicographic order
consistent with the orientation of lines in Figure~\ref{fig:T}.
Using the layer transfer matrix, the partition function of the model
can be expressed as
\begin{equation}
  Z
  =
  \Tr_{\bigotimes_{l,m} V_{(1_l 2_m)}}\biggl(\prod_n T^{[n]}(r_{3_n})\biggr)
  \,.
\end{equation}

Now, suppose that $R$, $\dot{R}$, $\ddot{R}$,
$\dddot{R} \colon \Z_2 \times C^3 \to \End(V^{\otimes 3})$ satisfy the
\emph{bicolored tetrahedron equations}
\begin{multline}
  \label{eq:BTE}
  R^{[\sigma]}_{(123)}(r_1,r_2,r_3)
  \dot{R}^{[\sigma+1]}_{(124)}(r_1,r_2,r_4)
  \ddot{R}^{[\sigma]}_{(134)}(r_1,r_3,r_4)
  \dddot{R}^{[\sigma+1]}_{(234)}(r_2,r_3,r_4)
  \\
  =
  \dddot{R}^{[\sigma]}_{(234)}(r_2,r_3,r_4)
  \ddot{R}^{[\sigma+1]}_{(134)}(r_1,r_3,r_4)
  \dot{R}^{[\sigma]}_{(124)}(r_1,r_2,r_4)
  R^{[\sigma+1]}_{(123)}(r_1,r_2,r_3)
  \,.
\end{multline}
We refer to \eqref{eq:BTE} as BTE[$\sigma$].

\begin{lemma}
  If $\dddot{R}$ is invertible, then $\bfR$, $\dot\bfR$, $\ddot\bfR$
  satisfy the \emph{bicolored Yang--Baxter equation}
  \begin{equation}
    \label{eq:BYBE}
    \bfR^{[\sigma]}_{23}
    \dot{\bfR}^{[\sigma+1]}_{24}
    \ddot{\bfR}^{[\sigma]}_{34}
    =
    \ddot{\bfR}^{[\sigma+1]}_{34}
    \dot{\bfR}^{[\sigma]}_{24}
    \bfR^{[\sigma+1]}_{23}
    \,.
  \end{equation}
\end{lemma}

\begin{proof}
  The left-hand side of \eqref{eq:BYBE} is equal to
  \begin{equation}
    \begin{aligned}
      &
      \Tr_{V_{(23)}}\biggl(\prod_{l} R^{[\sigma+l]}_{(1_l23)}\biggr)
      \Tr_{V_{(24)}}\biggl(\prod_{l} \dot{R}^{[\sigma+1+l]}_{(1_l24)}\biggr)
      \Tr_{V_{(34)}}\biggl(\prod_{l} \ddot{R}^{[\sigma+l]}_{(1_l34)}\biggr)
      \\
      &=
      \Tr_{V_{(23)} \otimes V_{(24)} \otimes V_{(34)}}
      \biggl(
      \Bigl(
      \prod_{l=1}^{2L}
      R^{[\sigma+l]}_{(1_l23)}
      \dot{R}^{[\sigma+1+l]}_{(1_l24)}
      \ddot{R}^{[\sigma+l]}_{(1_l34)}
      \Bigr)
      R^{[\sigma]}_{(234)}
      (R^{[\sigma]}_{(234)})^{-1}
      \biggr)
      \,.
    \end{aligned}
  \end{equation}
  By repeated use of the BTEs, we can rewrite it as
  \begin{equation}
    \label{eq:bYBE}
    \begin{aligned}
      &
      \Tr_{V_{(23)} \otimes V_{(24)} \otimes V_{(34)}}
      \biggl(
      R^{[\sigma]}_{(234)}
      \Bigl(
      \prod_{l}
      \ddot{R}^{[\sigma+l]}_{(1_l34)}
      \dot{R}^{[\sigma+1+l]}_{(1_l23)}
      R^{[\sigma+l]}_{(1_l23)}
      \Bigr)
      (R^{[\sigma]}_{(234)})^{-1}
      \biggr)
      \\
      &=
      \Tr_{V_{(23)} \otimes V_{(24)} \otimes V_{(34)}}
      \biggl(
      \prod_{l}
      \ddot{R}^{[\sigma+l]}_{(1_l34)}
      \dot{R}^{[\sigma+1+l]}_{(1_l23)}
      R^{[\sigma+l]}_{(1_l23)}
      \biggr)
      \\
      &=
      \Tr_{V_{(34)}}\biggl(\prod_{l} \ddot{R}^{[\sigma+1+l]}_{(1_l34)}\biggr)
      \Tr_{V_{(24)}}\biggl(\prod_{l} \dot{R}^{[\sigma+l]}_{(1_l24)}\biggr)
      \Tr_{V_{(23)}}\biggl(\prod_{l} R^{[\sigma+1+l]}_{(1_l23)}\biggr)
      \,,
    \end{aligned}
  \end{equation}
  which is the right-hand side of \eqref{eq:BYBE}.
\end{proof}

\begin{proposition}
  \label{prop:BTE}
  If $\dddot{R}$ and $\ddot{\bfR}$ are invertible, the layer transfer
  matrices $T$ of $R$ and $\dot{T}$ of $\dot{R}$ satisfy
  \begin{equation}
    \label{eq:TdTd=TT}
    T^{[\sigma]}(r) \dot{T}^{[\sigma+1]}(r')
    = \dot{T}^{[\sigma]}(r') T^{[\sigma+1]}(r) \,.
  \end{equation}
\end{proposition}

\begin{proof}
  Relation \eqref{eq:TdTd=TT} follows from the bicolored Yang--Baxter
  equation~\eqref{eq:bYBE}:
  \begin{equation}
    \begin{aligned}
      &
      \Tr_{\bfV_{3_{n}}}\biggl(\prod_m \bfR^{[\sigma+m]}_{2_m 3_n}\biggr)
      \Tr_{\bfV_{4_{n'}}}\biggl(\prod_m \dot{\bfR}^{[\sigma+1+m]}_{2_m 4_{n'}}\biggr)
      \\
      &=
      \Tr_{\bfV_{3_{n}} \otimes \bfV_{4_{n'}}}\biggl(
      \Bigl(
      \prod_m
      \bfR^{[\sigma+m]}_{2_m 3_n}
      \dot{\bfR}^{[\sigma+1+m]}_{2_m 4_{n'}}
      \Bigr)
      \ddot{\bfR}^{[\sigma]}_{3_n 4_{n'}}
      (\ddot{\bfR}^{[\sigma]}_{3_n 4_{n'}})^{-1}
      \biggr)
      \\
      &=
      \Tr_{\bfV_{3_{n}} \otimes \bfV_{4_{n'}}}\biggl(
      \ddot{\bfR}^{[\sigma]}_{3_n 4_{n'}}
      \Bigl(
      \prod_m
      \dot{\bfR}^{[\sigma+m]}_{2_m 4_{n'}}
      \bfR^{[\sigma+1+m]}_{2_m 3_n}
      \Bigr)
      (\ddot{\bfR}^{[\sigma]}_{3_n 4_{n'}})^{-1}
      \biggr)
      \\
      &=
      \Tr_{\bfV_{4_{n'}}}\biggl(\prod_m \dot{\bfR}^{[\sigma+m]}_{2_m 4_{n'}}\biggr)
      \Tr_{\bfV_{3_{n}}}\biggl(\prod_m \bfR^{[\sigma+1+m]}_{2_m 3_n}\biggr)
      \,.
    \end{aligned}
  \end{equation}
\end{proof}

Let
\begin{equation}
  \label{eq:twolayerT}
  \mathbf{T}(r, r') := T^{[0]}(r) T^{[1]}(r')
\end{equation}
be the combined transfer matrix for two adjacent layers.  By
Proposition~\ref{prop:BTE}, we have
\begin{equation}
  [\mathbf{T}(r, r'), \dot{\mathbf{T}}(s, s')] = 0 \,.
\end{equation}
This commutativity implies that the expansion of $\mathbf{T}$ and
$\dot{\mathbf{T}}$ in powers of spectral parameters produces two
families of operators such that any operator from one family commutes
with any operator from the other family.  This is different from the
integrability of the model defined by $R$ (or $\dot{R}$), unless the
two families coincide.

The following lemma provides a sufficient condition for the model to
be integrable:
\begin{lemma}
  \label{lemma:one-dimensional eigenspace2}
  
  If there exists $(t, t') \in C^2$ such that $\dot{\mathbf{T}}(t,t')$
  is diagonalizable with nondegenerate eigenvalues, then
  $[\mathbf{T}(r,r'), \mathbf{T}(s,s')] =0$ for any $(r,r')$,
  $(s,s') \in C^2$.
\end{lemma}
\begin{proof}
  Let $\ket{\Psi}$ be an eigenvector of $\dot{\mathbf{T}}(t,t')$ with
  eigenvalue $\lambda$.  Since $\dot{\mathbf{T}}(t,t')$ commutes with
  $\mathbf{T}(r,r')$, the vector $\mathbf{T}(r,r') \ket{\Psi}$ is zero
  or an eigenvector of $\dot{\mathbf{T}}(t,t')$ with eigenvalue
  $\lambda$.  By the assumption of nondegeneracy,
  $\mathbf{T}(r,r') \ket{\Psi}$ is a scalar multiple of $\ket{\Psi}$.
  Similarly, $\mathbf{T}(s,s') \ket{\Psi}$ is proportional to
  $\ket{\Psi}$, hence
  $[\mathbf{T}(r,r'), \mathbf{T}(s,s')] \ket{\Psi} = 0$.  By
  assumption, there is a basis consisting of eigenvectors of
  $\dot{\mathbf{T}}(t,t')$.  It follows that
  $[\mathbf{T}(r,r'), \mathbf{T}(s,s')] = 0$.
\end{proof}

\section{State integral models and the BTEs}
\label{sec:SIM-BTE}

\subsection{State integral models on shaped triangulations}

Consider a tetrahedron with totally ordered vertices $0$, $1$, $2$,
$3$, with each edge directed from a smaller to a larger vertex.  Up to
even permutations, there are two equivalence classes of vertex
orderings, corresponding to the two possible orientations of the
tetrahedron, which we call \emph{positive} and \emph{negative}.  We
say that the tetrahedron is given the shape of an ideal hyperbolic
tetrahedron if to each edge $e$ is assigned a positive real number
$\alpha(e)$, called the \emph{dihedral angle of $e$}, such that
\begin{equation}
  \alpha(01) = \alpha(23) =: \alpha \,,
  \quad
  \alpha(02) = \alpha(13) =: \beta \,,
  \quad
  \alpha(03) = \alpha(12) =: \gamma
\end{equation}
and
\begin{equation}
  \alpha + \beta + \gamma = \pi \,.
\end{equation}
See Figure~\ref{fig:tetrahedron} for an illustration of dihedral angles
and our convention for positive and negative orientations.

\begin{figure}
  \small
  \begin{equation*} 
  \begin{tikzpicture}[3d view={-40}{20}, scale=1.5, rotate=0, font=\scriptsize]
    \coordinate (0) at (0,0,1.2);
    \coordinate (3) at (1,0,0);
    \coordinate (1) at (-1/2,-{sqrt(3)/2},0); 
    \coordinate (2) at (-1/2,{sqrt(3)/2},0);
    
    \draw [->-=0.8] (2) -- node[pos=-0.05] {2} node[xshift=-4pt, sloped, above=-2pt] {$\alpha$} (3) node[pos=1.05] {3};

    \fill[black!10, opacity=0.6] (0) -- (3) -- (1) -- (2) -- cycle;

    \draw [->-=0.8] (0) node[above] {0}
    -- node[xshift=-4pt, sloped, above=-2pt] {$\alpha$} (1) node[pos=1.05] {1};
    \draw [->-=0.8] (0)
    -- node[xshift=0pt, sloped, above=-2pt] {$\beta$} (2);
    \draw [->-=0.8] (0)
    -- node[xshift=0pt, sloped, above=-2pt] {$\gamma$} (3);
    \draw [->-=0.8] (1) --  node[xshift=0pt, sloped, below=-2pt] {$\gamma$} (2);
    \draw [->-=0.8] (1)
    -- node[xshift=-4pt, sloped, below=-2pt] {$\beta$} (3);
  \end{tikzpicture}
  \qquad
  \begin{tikzpicture}[3d view={-40}{20}, scale=1.5, rotate=0, font=\scriptsize]
    \coordinate (0) at (0,0,1.2);
    \coordinate (2) at (1,0,0);
    \coordinate (1) at (-1/2,-{sqrt(3)/2},0); 
    \coordinate (3) at (-1/2,{sqrt(3)/2},0);
    
    \draw [->-=0.8] (2) -- node[pos=-0.05] {2} node[xshift=-4pt, sloped, above=-2pt] {$\alpha$} (3) node[pos=1.05] {3};

    \fill[black!10, opacity=0.6] (0) -- (3) -- (1) -- (2) -- cycle;

    \draw [->-=0.8] (0) node[above] {0}
    -- node[xshift=-4pt, sloped, above=-2pt] {$\alpha$} (1) node[pos=1.05] {1};
    \draw [->-=0.8] (0)
    -- node[xshift=0pt, sloped, above=-2pt] {$\beta$} (2);
    \draw [->-=0.8] (0)
    -- node[xshift=0pt, sloped, above=-2pt] {$\gamma$} (3);
    \draw [->-=0.8] (1) --  node[xshift=0pt, sloped, below=-2pt] {$\gamma$} (2);
    \draw [->-=0.8] (1)
    -- node[xshift=-4pt, sloped, below=-2pt] {$\beta$} (3);
  \end{tikzpicture}
  \end{equation*}

  \caption{The tetrahedral weights for a positive tetrahedron (left)
    and a negative tetrahedron (right).}
  \label{fig:tetrahedron}
\end{figure}
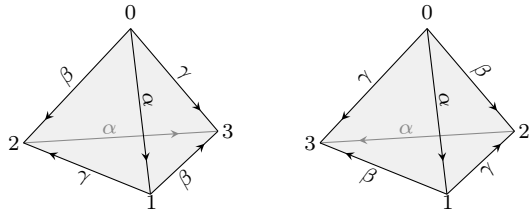

An \emph{oriented triangulated pseudo $3$-manifold} is a geometric
object constructed from finitely many tetrahedra with ordered vertices
by gluing them along faces while matching the directions of identified
edges.  The set of the tetrahedra glued together provides its
\emph{triangulation}.  A \emph{shape structure} of an oriented
triangulated pseudo $3$-manifold is an assignment of the shape of an
ideal hyperbolic tetrahedron to each tetrahedron in the triangulation.
An oriented pseudo $3$-manifold endowed with a shape structure is
called a \emph{shaped pseudo $3$-manifold}.  An internal edge of a
shaped pseudo $3$-manifold is \emph{balanced} if the total angle
around that edge (that is, the sum of the dihedral angles of the edge
in all tetrahedra containing it) is equal to $2\pi$.



A \emph{shape gauge transformation} with parameter $\theta$ performed
on edge $0v$ ($v \in \{1,2,3\}$) or $xy$ ($x$, $y \notin \{0,v\}$) of
a tetrahedron with ordered vertices $0$, $1$, $2$, $3$ is the shift of
dihedral angles
\begin{equation}
  \alpha(0w) \to \alpha(0w) \pm \theta \sum_{x=1}^3 \epsilon_{vwx} \,,
\end{equation}
where $\epsilon$ is a completely antisymmetric tensor with
$\epsilon_{123} = 1$ and the sign is determined by the orientation of
the tetrahedron.  A shape gauge transformation performed on an edge of
a shaped pseudo $3$-manifold is defined by the same action on all
tetrahedra containing that edge in the triangulation.  Shape gauge
transformations leave the total angles of internal edges invariant.

Consider a shaped bipyramid consisting of two tetrahedra.  A
\emph{shaped $2$--$3$ move} applied to this bipyramid is a change of
shaped triangulation to another one consisting of three tetrahedra in
a way that the total angle around each external edge is preserved
and the induced internal edge is balanced.%
\footnote{Shaped $2$--$3$ moves are bidirectional, but in some
  contexts it is convenient to reserve this term to be used for the
  direction from two to three tetrahedra, distinguishing them from
  ``$3$--$2$ moves'' which go in the opposite direction.}
For a shaped $2$--$3$ move to be applicable, there must be a total
ordering among the five vertices of the bipyramid consistent with the
directions of edges; any pair of vertices is contained in the same
tetrahedron either before or after the move is applied.  See
Figure~\ref{fig:2-3} for an illustration of a shaped $2$--$3$ move.
There is more than one shaped $2$--$3$ move that can be applied to a
given bipyramid, differing by shape gauge transformations on the
internal edge.

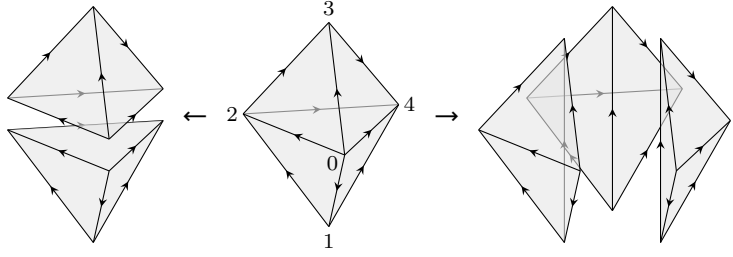
\begin{figure}
\begin{equation*}
  \begin{tikzpicture}[3d view={-40}{20}, scale=1.2, rotate=0, font=\scriptsize]
  \begin{scope}[yshift=-5pt]
    \coordinate (0) at (-1/2,-{sqrt(3)/2},0); 
    \coordinate (1) at (0,0,-1.2);
    \coordinate (2) at (-1/2,{sqrt(3)/2},0);
    \coordinate (3) at (0,0,1.2);
    \coordinate (4) at (1,0,0);
    
    \draw [->-=0.5] (2) -- (4);
    \fill[black!10, opacity=0.6] (1) -- (2) -- (4) -- cycle;
    \draw [->-=0.5] (0) -- (1);
    \draw [->-=0.5] (0) -- (2);
    \draw [->-=0.5] (0) -- (4);
    \draw [->-=0.5] (1) -- (2);
    \draw [->-=0.5] (1) -- (4);
  \end{scope}
  \begin{scope}[yshift=5pt]
    \coordinate (0) at (-1/2,-{sqrt(3)/2},0); 
    \coordinate (1) at (0,0,-1.2);
    \coordinate (2) at (-1/2,{sqrt(3)/2},0);
    \coordinate (3) at (0,0,1.2);
    \coordinate (4) at (1,0,0);
    
    \draw [->-=0.5] (2) -- (4);
    \fill[black!10, opacity=0.6] (0) -- (2) -- (3) -- (4) -- cycle;
    \draw [->-=0.5] (0) -- (2);
    \draw [->-=0.5] (0) -- (3);
    \draw [->-=0.5] (0) -- (4);
    \draw [->-=0.5] (2) -- (3);
    \draw [->-=0.5] (3) -- (4);
  \end{scope}
\end{tikzpicture}
\ \leftarrow
\begin{tikzpicture}[3d view={-40}{20}, scale=1.2, rotate=0, font=\scriptsize]
  \coordinate (0) at (-1/2,-{sqrt(3)/2},0); 
  \coordinate (1) at (0,0,-1.2);
  \coordinate (2) at (-1/2,{sqrt(3)/2},0);
  \coordinate (3) at (0,0,1.2);
  \coordinate (4) at (1,0,0);
  
  \draw [->-=0.5] (2) -- node[pos=-0.07] {$2$} (4) node[pos=1.07] {$4$};
  \fill[black!10, opacity=0.6] (1) -- (2) -- (3) -- (4) -- cycle;
  \draw [->-=0.5] (0) -- (1) node[below=-1pt] {$1$};
  \draw [->-=0.5] (0) -- (2);
  \draw [->-=0.5] (0) node[xshift=-4.5pt, yshift=-3pt] {$0$} -- (3) node[above=-1pt] {$3$};
  \draw [->-=0.5] (0) -- (4);
  \draw [->-=0.5] (1) -- (2);
  \draw [->-=0.5] (1) -- (4);
  \draw [->-=0.5] (2) -- (3);
  \draw [->-=0.5] (3) -- (4);
\end{tikzpicture}
\rightarrow\
\begin{tikzpicture}[3d view={-40}{20}, scale=1.2, rotate=0, font=\scriptsize]
  \begin{scope}[yshift=10pt]
    \coordinate (0) at (-1/2,-{sqrt(3)/2},0); 
    \coordinate (1) at (0,0,-1.2);
    \coordinate (2) at (-1/2,{sqrt(3)/2},0);
    \coordinate (3) at (0,0,1.2);
    \coordinate (4) at (1,0,0);

    \draw [->-=0.5] (2) -- (4);
    \fill[black!10, opacity=0.6] (1) -- (2) -- (3) -- (4) -- cycle;
    \draw [->-=0.5] (1) -- (2);
    \draw [->-=0.5] (1) -- (4);
    \draw [->-=0.5] (1) -- (3);
    \draw [->-=0.5] (2) -- (3);
    \draw [->-=0.5] (3) -- (4);
  \end{scope}

  \begin{scope}[xshift=15pt]
    \coordinate (0) at (-1/2,-{sqrt(3)/2},0); 
    \coordinate (1) at (0,0,-1.2);
    \coordinate (2) at (-1/2,{sqrt(3)/2},0);
    \coordinate (3) at (0,0,1.2);
    \coordinate (4) at (1,0,0);

    \fill[black!10, opacity=0.6] (1) -- (3) -- (4) -- cycle;

    \draw [->-=0.5] (0) -- (1);
    \draw [->-=0.5] (0) -- (3);
    \draw [->-=0.5] (0) -- (4);
    \draw [->-=0.5] (1) -- (4);
    \draw [->-=0.5] (1) -- (3);
    \draw [->-=0.5] (3) -- (4);
  \end{scope}

  \begin{scope}[xshift=-15pt]
    \coordinate (0) at (-1/2,-{sqrt(3)/2},0); 
    \coordinate (1) at (0,0,-1.2);
    \coordinate (2) at (-1/2,{sqrt(3)/2},0);
    \coordinate (3) at (0,0,1.2);
    \coordinate (4) at (1,0,0);

    \draw [->-=0.5] (1) -- (3);
    \fill[black!10, opacity=0.6] (1) -- (2) -- (3) -- (0) -- cycle;
    \draw [->-=0.5] (0) -- (1);
    \draw [->-=0.5] (0) -- (2);
    \draw [->-=0.5] (0) -- (3);
    \draw [->-=0.5] (1) -- (2);
    \draw [->-=0.5] (2) -- (3);
  \end{scope}
\end{tikzpicture}
\end{equation*}

\caption{An example of a shaped $2$--$3$ move.}
  \label{fig:2-3}
\end{figure}

We will say two shaped pseudo $3$-manifolds are \emph{equivalent} if
they are related by a sequence of shaped $2$--$3$ moves and shape gauge
transformations on internal edges.

For the purpose of the present work, by a \emph{state integral model
  on shaped triangulations} we mean a statistical mechanics model
defined on shaped pseudo $3$-manifolds such that its partition
functions on equivalent shaped pseudo $3$-manifolds are equal.  Such a
model assigns a local Boltzmann weight to a configuration of state
variables located on a single tetrahedron.  Typically, state variables
are continuous and located on the edges and faces.  State variables on
internal edges and faces are integrated over.

\subsection{BTEs from state integral models}

We wish to solve the BTEs geometrically by representing the eight
R-matrices $R^{[\sigma]}$, $\dot{R}^{[\sigma]}$,
$\ddot{R}^{[\sigma]}$, $\dddot{R}^{[\sigma]}$, $\sigma = 0$, $1$, with
shaped cubes.  From this geometric point of view, each of the BTEs is
a statement that two shaped rhombic dodecahedra, obtained from two
sets of four shaped cubes glued together, are equivalent.  See
Figure~\ref{fig:RD}.

\tdplotsetmaincoords{-70}{-52}

We choose to triangulate each cube by six tetrahedra as in
Figure~\ref{fig:cubes}.  (The meaning of line thickness will be
explained shortly.) For this choice, the rhombic dodecahedra appearing
in BTE[0] and BTE[1] consist of the cubes shown in
Figure~\ref{fig:BTE0-cubes} and Figure~\ref{fig:BTE1-cubes}, respectively.

\def\arraystretch{1.8}

\begin{figure}
  \centering
\begin{equation*}
  \begin{array}{cccc}
    \begin{tikzpicture}[tdplot_main_coords]
      \cube
      \Rzero
    \end{tikzpicture}
    &
      \begin{tikzpicture}[tdplot_main_coords]
        \cube
        \Rdzero
      \end{tikzpicture}
    &
      \begin{tikzpicture}[tdplot_main_coords]
        \cube
        \Rddzero
      \end{tikzpicture}
    &
      \begin{tikzpicture}[tdplot_main_coords]
        \cube
        \Rdddzero
      \end{tikzpicture}
    \\
    R^{[0]} & \dot{R}^{[0]}
    & \ddot{R}^{[0]} & \dddot{R}^{[0]}
    \\
    \begin{tikzpicture}[tdplot_main_coords]
      \cube    
      \Rone
    \end{tikzpicture}
    &
      \begin{tikzpicture}[tdplot_main_coords]
        \cube
        \Rdone
      \end{tikzpicture}
    &
      \begin{tikzpicture}[tdplot_main_coords]
        \cube
        \Rddone
      \end{tikzpicture}
    &
      \begin{tikzpicture}[tdplot_main_coords]
        \cube
        \Rdddone
      \end{tikzpicture}
    \\
    R^{[1]} & \dot{R}^{[1]}
    & \ddot{R}^{[1]} & \dddot{R}^{[1]}
  \end{array}
\end{equation*}

\caption{The shaped triangulations of cubes}
\label{fig:cubes}
\end{figure}
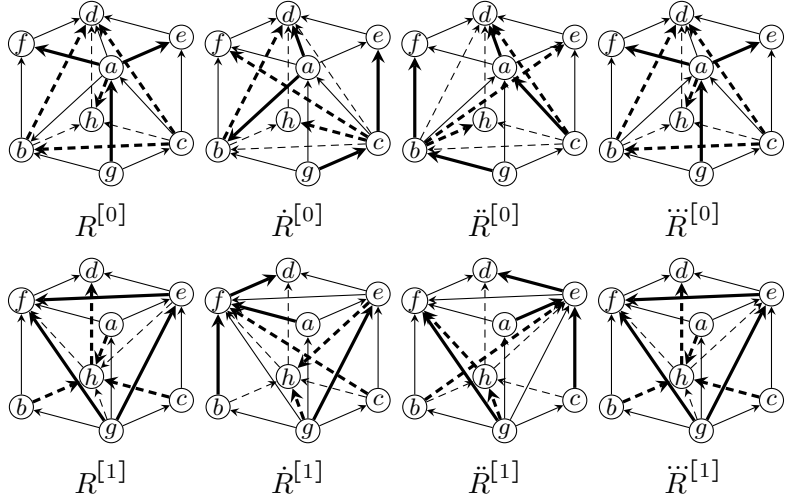

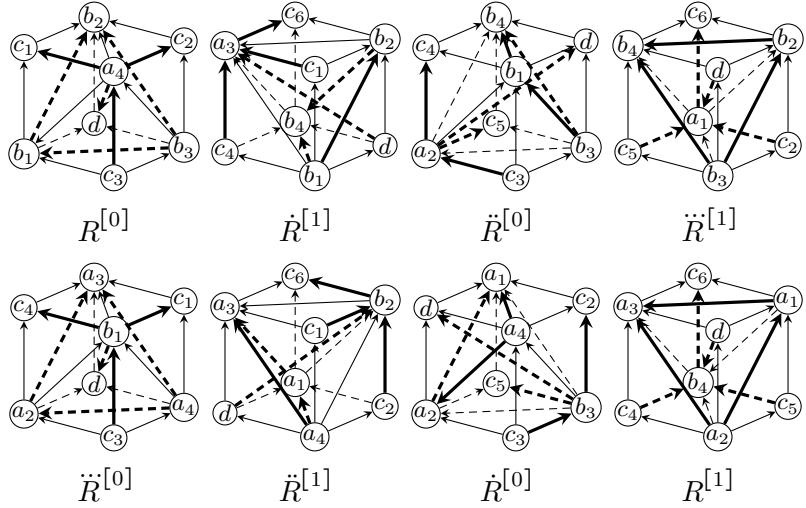
\begin{figure}
  \centering
  \begin{equation*}
    \begin{array}{cccc}
      \begin{tikzpicture}[tdplot_main_coords]
        \TEcubeLHSone
        \Rzero
      \end{tikzpicture}
      &
        \begin{tikzpicture}[tdplot_main_coords]
          \TEcubeLHStwo    
          \Rdone
        \end{tikzpicture}
      &
        \begin{tikzpicture}[tdplot_main_coords]
          \TEcubeLHSthree
          \Rddzero
        \end{tikzpicture}
      &
        \begin{tikzpicture}[tdplot_main_coords]
          \TEcubeLHSfour
          \Rdddone
        \end{tikzpicture}
      \\
      R^{[0]} & \dot{R}^{[1]} & \ddot{R}^{[0]} & \dddot{R}^{[1]}
      \\
      \begin{tikzpicture}[tdplot_main_coords]
        \TEcubeRHSone
        \Rdddzero
      \end{tikzpicture}
      &
        \begin{tikzpicture}[tdplot_main_coords]
          \TEcubeRHStwo
          \Rddone
        \end{tikzpicture}
      &
        \begin{tikzpicture}[tdplot_main_coords]
          \TEcubeRHSthree
          \Rdzero
        \end{tikzpicture}
      &
        \begin{tikzpicture}[tdplot_main_coords]
          \TEcubeRHSfour
          \Rone
        \end{tikzpicture}
      \\
      \dddot{R}^{[0]} & \ddot{R}^{[1]} & \dot{R}^{[0]} &  R^{[1]}
    \end{array}
  \end{equation*}
  
  \caption{The cubes for the shaped rhombic dodecahedra on the left-hand
    side (top) and the right-hand side (bottom) of BTE[0].}
  \label{fig:BTE0-cubes}
\end{figure}

\begin{figure}
  \centering
  \begin{equation*}
    \begin{array}{cccc}
      \begin{tikzpicture}[tdplot_main_coords]
        \TEcubeLHSone
        \Rone
      \end{tikzpicture}
      &
        \begin{tikzpicture}[tdplot_main_coords]
          \TEcubeLHStwo    
          \Rdzero
        \end{tikzpicture}
      &
        \begin{tikzpicture}[tdplot_main_coords]
          \TEcubeLHSthree
          \Rddone
        \end{tikzpicture}
      &
        \begin{tikzpicture}[tdplot_main_coords]
          \TEcubeLHSfour
          \Rdddzero
        \end{tikzpicture}
      \\
      R^{[1]} & \dot{R}^{[0]} & \ddot{R}^{[1]} & \dddot{R}^{[0]}
      \\
      \begin{tikzpicture}[tdplot_main_coords]
        \TEcubeRHSone
        \Rdddone
      \end{tikzpicture}
      &
        \begin{tikzpicture}[tdplot_main_coords]
          \TEcubeRHStwo
          \Rddzero
        \end{tikzpicture}
      &
        \begin{tikzpicture}[tdplot_main_coords]
          \TEcubeRHSthree
          \Rdone
        \end{tikzpicture}
      &
        \begin{tikzpicture}[tdplot_main_coords]
          \TEcubeRHSfour
          \Rzero
        \end{tikzpicture}
      \\
      \dddot{R}^{[1]} & \ddot{R}^{[0]} & \dot{R}^{[1]} &  R^{[0]}
    \end{array}
  \end{equation*}

  \caption{The cubes for the shaped rhombic dodecahedra on the left-hand
    side (top) and the right-hand side (bottom) of BTE[1].}
  \label{fig:BTE1-cubes}
\end{figure}
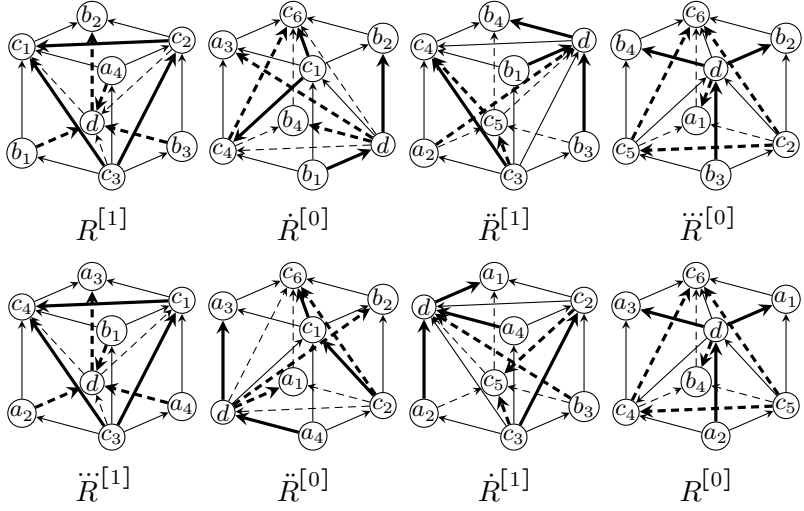

We must find shape structures for these triangulated cubes for which
the BTEs hold.  One may hope to make all body diagonal edges (such as
edge $ah$ in $R^{[0]}$) balanced so that for each $\sigma$, all
$R^{[\sigma]}$, $\dot{R}^{[\sigma]}$, $\ddot{R}^{[\sigma]}$,
$\dddot{R}^{[\sigma]}$ can be transformed by shaped 2-3 moves to have
the same shaped triangulation with five tetrahedra.  Unfortunately,
this is not possible.  Looking at the tetrahedra containing the body
diagonal edges on the left-hand side of BTE[0], we deduce
\begin{equation}
  \label{eq:omega-da}
  \sum_{i=1}^4 \omega(da_i) = 12\pi - \sum_{i=1}^4 \omega(db_i) \,,
\end{equation}
where $\omega(da_i)$ and $\omega(db_i)$ denote the total angle around
the edges $da_i$ and $db_i$.  The edges $db_1$, $db_2$, $db_3$,
$db_4$, however, are body diagonal on the right-hand side of BTE[0].
Therefore, some body diagonal edges must be unbalanced.

The constraint \eqref{eq:omega-da} is satisfied if we take the total
angles around the body diagonal edges to be all equal to $3\pi/2$.  A
symmetric choice is the following angle assignment which uses only
tetrahedra whose angles are $\pi/2$, $\pi/4$, $\pi/4$.  As we will see
shortly, the BTEs are satisfied if we choose, for each tetrahedron,
the edges drawn by thick lines in Figure~\ref{fig:cubes} to have
dihedral angle $\pi/2$.

In order to introduce parameters to the R-matrices, we can make use of
shape gauge transformations.  To the eight cubes in BTE[0], let us
apply the shape gauge transformations with parameters
\begin{itemize}
\item $s_i$ on the edges parallel to $da_i$ and $db_i$;
\item $t_{ij}$, $i < j$, on $a_ia_j$ and $b_ib_j$; and
\item $u_{ij}$, $i < j$, $\{i,j\} \sqcup \{k, l\} = \{1,2,3,4\}$,
  on $a_kb_l$ and $a_lb_k$.
\end{itemize}
This results in deformed shaped cubes that define R-matrices depending
on parameters.  For example, $R^{[0]}$ is deformed by the shaped gauge
transformations with parameters assigned to the edges as in
Figure~\ref{fig:parameters}.  The deformed $R^{[0]}$ is a function of
the nine parameters $s_1$, $s_2$, $s_3$, $t_{12}$, $t_{13}$, $t_{23}$,
$u_{12}$, $u_{13}$, $u_{23}$; it is independent of $s_4$ since this
parameter is assigned to the internal edge $ah$.

\begin{figure}
  \centering
  \begin{tikzpicture}[tdplot_main_coords, scale=1.5]
    \begin{scope}[scale=1.5]
      \node[gnode, opacity=0.2] (b) at (0,1,1) {$b$};
      \node[gnode, opacity=0.2] (c) at (1,0,1) {$c$};
      \node[gnode, opacity=0.2] (d) at (1,1,0) {$d$};
      \node[gnode, opacity=0.2] (h) at (1,1,1) {$h$};
      
      \node[gnode, opacity=0.2] (a) at (0,0,0) {$a$};
      \node[gnode, opacity=0.2] (e) at (1,0,0) {$e$};
      \node[gnode, opacity=0.2] (f) at (0,1,0) {$f$};
      \node[gnode, opacity=0.2] (g) at (0,0,1) {$g$};
    \end{scope}
    
    \draw[opacity=0.2, ->] (a) -- node[text opacity=1, sloped] {$s_1$} (e);
    \draw[opacity=0.2, ->] (a) -- node[text opacity=1, sloped] {$s_3$} (f);
    \draw[opacity=0.2, ->] (g) -- node[text opacity=1, sloped] {$s_2$} (a);
    \draw[opacity=0.2, ->] (c) -- node[text opacity=1, sloped] {$s_2$} (e);
    \draw[opacity=0.2, ->] (e) -- node[text opacity=1, sloped] {$s_3$} (d);
    \draw[opacity=0.2, ->] (f) -- node[text opacity=1, sloped] {$s_1$} (d);
    \draw[opacity=0.2, ->] (b) -- node[text opacity=1, sloped] {$s_2$} (f);
    \draw[opacity=0.2, ->] (g) -- node[text opacity=1, sloped] {$s_3$} (b);
    \draw[opacity=0.2, ->] (g) -- node[text opacity=1, sloped] {$s_1$} (c);
    
    \draw[opacity=0.2, densely dashed, ->] (b) -- node[text opacity=1,
    sloped] {$s_1$} (h); \draw[opacity=0.2, densely dashed, ->] (c) --
    node[text opacity=1, sloped] {$s_3$} (h); \draw[opacity=0.2,
    densely dashed, ->] (h) -- node[text opacity=1, sloped] {$s_2$}
    (d);
    
    \draw[opacity=0.2, ->] (a) -- node[text opacity=1, sloped] {$u_{23}$} (b);
    \draw[opacity=0.2, ->] (c) -- node[text opacity=1, sloped] {$u_{12}$} (a);
    \draw[opacity=0.2, ->] (a) -- node[text opacity=1, sloped] {$u_{13}$} (d);
    \draw[opacity=0.2, densely dashed, ->] (b) -- node[text opacity=1, sloped] {$t_{12}$} (d);
    \draw[opacity=0.2, densely dashed, ->] (c) -- node[text opacity=1, sloped] {$t_{23}$} (d);
    \draw[opacity=0.2, densely dashed, ->] (c) -- node[text opacity=1, sloped] {$t_{13}$} (b);
    \draw[opacity=0.2, densely dashed, ->] (a) -- node[text opacity=1, sloped] {$s_4$} (h);
  \end{tikzpicture}

\caption{The parameter assignment for $R^{[0]}$.}
\label{fig:parameters}
\end{figure}
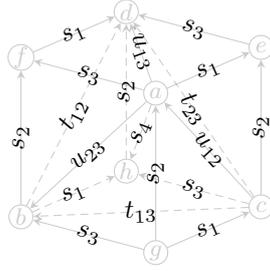

\begin{proposition}
  \label{prop:BTEs-cubes}
  For generic values of parameters, the above R-matrices satisfy the
  BTEs.
\end{proposition}

\begin{proof}
  For BTE[1], the rhombic dodecahedra on the two sides are equivalent
  because their shaped triangulations differ only by shape gauge
  transformations on internal edges.


  For BTE[0], each rhombic dodecahedron consists of six shaped
  octahedra whose vertices are sets of the form
  $\{a_i, a_j, b_k, b_l, c_m, d\}$, with
  $\{i,j\} \sqcup \{k, l\} = \{1,2,3,4\}$.
  The corresponding octahedra on the two sides are related by two
  shaped $2$--$3$ moves, hence the two rhombic dodecahedra are
  equivalent.

  To be more concrete, let us consider the octahedron consisting of
  four tetrahedra with vertices $\{a_1, b_3, b_4, c_5\}$,
  $\{a_2, b_3, b_4, c_5\}$, $\{a_1, b_3, b_4, d\}$ and
  $\{a_2, b_3, b_4, d\}$ in the rhombic dodecahedron on the left-hand
  side.  Depending on whether $u_{13} + u_{24} - u_{23} -u_{14}$ is
  positive or negative, we can apply the $2$--$3$ move to the
  bipyramid consisting of the first two tetrahedra or the last two
  tetrahedra.  Consider the former case.  The $2$--$3$ move produces
  three tetrahedra with vertices $\{a_1, a_2, b_3, b_4\}$,
  $\{a_1, a_2, b_3, c_5\}$, $\{a_1, a_2, b_4, c_5\}$.  Then, we apply
  the $3$--$2$ move to the bipyramid consisting of
  $\{a_1, a_2, b_3, b_4\}$, $\{a_1, b_3, b_4, d\}$,
  $\{a_2, b_3, b_4, d\}$ to obtain four tetrahedra
  $\{a_1, a_2, b_3, c_5\}$, $\{a_1, a_2, b_4, c_5\}$,
  $\{a_1, a_2, b_3, d\}$, $\{a_1, a_2, b_4, d\}$, which comprise the
  corresponding octahedron on the right-hand side.
\end{proof}

If we choose $t_{ij} = t(t_i, t_j)$, $u_{ij} = u(u_i, u_j)$ for some
functions $t$, $u$ of two variables, then the triplet
$r_i = (s_i, t_i, u_i)$ may be thought of as a parameter assigned to
the $i$-th plane.  However, it cannot be interpreted as a spectral
parameter because, as a gauge transformation parameter, it cannot
change the partition function on a periodic cubic lattice, unless the
boundary condition is modified to break shape gauge invariance.

\section{Solution by Teichm\"uller TQFT}
\label{sec:Teichmuller}

In this section, we discuss a solution to the BTEs produced by
Teichm\"uller TQFT, which is a state integral model introduced by
Andersen and Kashaev \cite{MR3227503} and built on earlier works,
including \cite{Hikami:2001en, MR2330673, MR3250765, MR2551896,
  Dijkgraaf:2010ur}.

\subsection{Teichm\"uller TQFT}

Let $\bb$ be a complex number such that $\Re\bb \neq 0$ and
\begin{equation}
  c_\bb := \frac{i}{2} (\bb+\bb^{-1}) \in i\R \,,
\end{equation}
and let $\pp_i$, $\qq_i$ be position and momentum operators obeying
the canonical commutation relations
\begin{equation}
  [\pp_i, \pp_j] = [\qq_i, \qq_j] = 0 \,,
  \qquad
  [\pp_i, \qq_j] = \frac{1}{2\pi i} \delta_{i,j} \,.
\end{equation}
Using Faddeev's noncompact quantum dilogarithm $\Phi_\bb$, define the
\emph{charged tetrahedral operator}
\begin{equation}
  \TT_{ij}(a,c)
  :=
  e^{2\pi i \pp_i\qq_j} \psi_{a,c}(\qq_i - \qq_j + \pp_j) \,,
\end{equation}
where $a$, $c \in \R_{>0}$ and
\begin{equation}
  \psi_{a,c}(x)
  :=
  \Phi_\bb\bigl(x - 2c_\bb(a+c)\bigr)^{-1}
  e^{-4\pi ic_\bb a(x - c_\bb(a+c))}
  e^{-\pi ic_\bb^2 (4(a-c)+1)/6} \,.
\end{equation}
It satisfies the identity
\begin{equation}
  \label{eq:CPI}
  \TT_{ij}(a_4,c_4)\TT_{ik}(a_2,c_2)\TT_{jk}(a_0,c_0)
  = e^{\pi ic_\bb^2 P_e/3} \TT_{jk}(a_1,c_1)\TT_{ij}(a_3,c_3) \,,
\end{equation}
where $P_e = 2(c_0+a_2+c_4) - 1/2$ and
\begin{equation}
  a_1 = a_0 + a_2 \,,
  \
  a_3 = a_2 + a_4 \,,
  \
  c_1 = c_0 + a_4 \,,
  \
  c_3 = a_0 + c_4 \,,
  \
  c_2 = c_1 + c_3 \,.
\end{equation}




In the original formulation of Teichm\"uller TQFT \cite{MR3227503},
state variables are located on the faces of shaped triangulations.
Let $x_v$, $v = 0$, $1$, $2$, $3$, be the state variable on the face
opposite to a vertex $v$ of a tetrahedron with ordered vertices and
dihedral angles $(\alpha, \beta, \gamma) = (2\pi a, 2\pi b, 2\pi c)$.
The Boltzmann weight assigned to this state configuration is given by
the expectation value
\begin{equation}
  \TT(a,c)_{x_1 x_3}^{x_0 x_2} := \bra{x_0, x_2} \TT(a,c) \ket{x_1, x_3}\,,
\end{equation}
or
\begin{equation}
  \TTb(a,c)_{x_0 x_2}^{x_1 x_3}
  :=
  \overline{\bra{x_0, x_2} \TT(a,c) \ket{x_1, x_3}} \,,
\end{equation}
depending on whether the tetrahedron is positive or negative.
Explicitly, we have%
\footnote{The R-matrix depends on the gauge parameter assigned to the
  body diagonal edges through a phase factor.  Such phase factors
  cancel in the BTEs and can be omitted from the definitions of the
  R-matrices.}
\begin{equation}
  \bra{x_0, x_2} \TT(a,c) \ket{x_1, x_3}
  =
  \delta(x_0+x_2-x_1)\psit'_{a,c}(x_3-x_2) e^{2\pi ix_0(x_3-x_2)} \,,
\end{equation}
where
\begin{equation}
  \psit'_{a,c}(x) := e^{-\pi ix^2} \int_\R \psi_{a,c}(y) e^{-2\pi ixy} dy \,.
\end{equation}

Teichm\"uller TQFT has the following properties:
\begin{itemize}
\item \emph{Shape gauge invariance}. Under a shape gauge
  transformation with parameter $\theta$ performed on an internal edge
  $e$ that has a total angle $\omega(e)$ shared by $n$ tetrahedra, the
  partition function gets multiplied by the factor
  \begin{equation}
    \label{eq:SGI-Teichmuller}
    e^{ic_\bb^2 \theta(n/3 - \omega(e)/\pi)} \,.
  \end{equation}

\item \emph{Tetrahedral symmetry}. There exist cones over bigons,
  which can be attached to a tetrahedron to change the vertex ordering.
  As a result, the partition function on a closed shaped pseudo
  $3$-manifold is independent of the choice of vertex ordering of
  tetrahedra.

\item \emph{Invariance under shaped $2$--$3$ moves}.  The identity
  \eqref{eq:CPI} expresses the invariance, up to a phase, of the
  partition function of the bipyramid in Figure~\ref{fig:2-3} under the
  shaped $2$--$3$ move.  By the tetrahedral symmetry, the invariance
  under $2$--$3$ moves holds for any vertex ordering of the bipyramid.
\end{itemize}
Therefore, the partition functions of Teichm\"uller TQFT on equivalent
shaped pseudo $3$-manifolds are equal up to a phase.

\subsection{Teichm\"uller TQFT solves the BTEs}

The R-matrices produced by Teichm\"uller TQFT carry real state
variables placed on the external faces of the shaped cubes.  As such,
they are of vertex type, with each edge carrying a pair of state
variables.

Explicitly, the matrix elements of $R^{[0]}$ are given by
\begin{equation}
  \begin{aligned}
    R^{[0]}&_{(x_3, \xb_3) (x_2, \xb_2) (x_1, \xb_1)}^{(x'_3, \xb'_3) (x'_2, \xb'_2) (x'_1, \xb'_1)}
    \\
    = \int & dy_1 \, dy_2 \, dy_3 \, dz_1 \, dz_2 \, dz_3
    \\
    &
    \TT\bigl(\tfrac18 + \tfrac{1}{2\pi}(s_2 - s_3 + u_{13} - t_{12}),
    \tfrac18 + \tfrac{1}{2\pi}(s_3 - s_1 + t_{12} - u_{23})\bigr)_{x_2 x_3}^{x'_1 y_1}
    \\
    \times &
    \TTb\bigl(\tfrac18 + \tfrac{1}{2\pi}(s_4 - s_1 + t_{12} - u_{13}),
    \tfrac18 + \tfrac{1}{2\pi}(s_2 - s_4 - t_{12} + u_{23})\bigr)_{\xb'_1 y_1}^{z_1 z_2}
    \\
    \times &
    \TTb\bigl(\tfrac18 + \tfrac{1}{2\pi}(s_2 - s_3 + t_{13} - u_{12}),
    \tfrac18 + \tfrac{1}{2\pi}(s_1 - s_2 - t_{13} + u_{23})\bigr)_{y_2 \xb'_2}^{\xb_3 \xb_1}
    \\
    \times &
    \TTb\bigl(\tfrac18 + \tfrac{1}{2\pi}(s_4 - s_3 + t_{13} - u_{23}),
    \tfrac18 + \tfrac{1}{2\pi}(s_1 - s_4 + u_{12} - t_{13})\bigr)_{z_2 z_3}^{x'_2 y_2}
    \\
    \times &
    \TT(\tfrac18 + \tfrac{1}{2\pi}(s_1 - s_2 + t_{23} - u_{13}),
    \tfrac14 + \tfrac{1}{2\pi}(s_2 - s_3 + u_{13} - u_{12})\bigr)_{\xb'_3 x_1}^{\xb_2 y_3}
    \\
    \times &
    \TTb\bigl(\tfrac18 + \tfrac{1}{2\pi}(s_3 - s_4 + u_{13} - t_{23}),
    \tfrac14 + \tfrac{1}{2\pi}(s_2 - s_3 + u_{12} - u_{13})\bigr)_{z_1 y_3}^{x'_3 z_3}
    \,.
  \end{aligned}
\end{equation}
This is proportional to the delta function
\begin{equation}
  \delta(\xb'_1 - x'_1 + \xb'_3  - x'_3 + x_2 - \xb_2) \,,
\end{equation}
which expresses the conservation law
$x_2 - \xb_2 = x'_1 - \xb'_1 + x'_3 - \xb'_3$. 

Similarly, the matrix elements of $R^{[1]}$ are given by
 \begin{equation}
  \begin{aligned}
    R^{[1]}&_{(x_3, \xb_3) (x_2, \xb_2) (x_1, \xb_1)}^{(x'_3, \xb'_3) (x'_2, \xb'_2) (x'_1, \xb'_1)}
    \\
    = \int & dy_1 \, dy_2 \, dy_3 \, dz_1 \, dz_2 \, dz_3
    \\
    &
    \TTb\bigl(\tfrac18 + \tfrac{1}{2\pi}(s_2-s_1+u_{13}-t_{23}),
    \tfrac18 + \tfrac{1}{2\pi}(s_3-s_2+u_{12}-u_{13})\bigr)^{y_2 \xb'_2}_{x'_1 x_3}
    \\
    \times &
    \TT\bigl(\tfrac18 + \tfrac{1}{2\pi}(s_4-s_3+t_{23}-u_{13}),
    \tfrac18 + \tfrac{1}{2\pi}(s_3-s_2 + u_{13}- u_{12})\bigr)^{z_1 \xb_3}_{y_2 z_3}
    \\
    \times &
    \TT\bigl(\tfrac18 + \tfrac{1}{2\pi}(s_3-s_2+u_{12}-t_{13}),
    \tfrac18 + \tfrac{1}{2\pi}(s_2-s_1+t_{13}-u_{23})\bigr)^{\xb'_1 x'_3}_{x_2 y_1}
    \\
    \times &
    \TT\bigl(\tfrac18 + \tfrac{1}{2\pi}(s_3-s_4+u_{23}-t_{13}),
    \tfrac18 + \tfrac{1}{2\pi}(s_4-s_1+t_{13}-u_{12})\bigr)^{y_1 \xb_2}_{z_1 z_2}
    \\
    \times &
    \TTb(\tfrac18 + \tfrac{1}{2\pi}(s_3-s_2+t_{12}-u_{13}),
    \tfrac14 + \tfrac{1}{2\pi}(s_1-s_3+u_{23}-t_{12})\bigr)^{y_3 x_1}_{\xb'_3 x'_2}
    \\
    \times &
    \TT\bigl(\tfrac18 + \tfrac{1}{2\pi}(s_1-s_4+u_{13}-t_{12}),
    \tfrac14 + \tfrac{1}{2\pi}(s_4-s_2+t_{12}-u_{23})\bigr)^{z_2 z_3}_{y_3 \xb_1}
    \,.
  \end{aligned}
\end{equation} 
This is proportional to the delta function
\begin{align}
    \delta(\xb'_1+x'_3-x_2)\,,
\end{align}
which expresses the conservation law $x_2=\xb'_1+x'_3$.

A priori, the R-matrices produced by Teichm\"uller TQFT are guaranteed
to solve the BTEs only up to a phase since the equality of partition
functions between equivalent shaped pseudo $3$-manifolds holds only up
to a phase.  Thus, the left-hand side of BTE[$\sigma$] is equal to the
right-hand side of BTE[$\sigma$], multiplied by a phase factor
$e^{i\Delta_\sigma}$ which may depend on the parameters $s_i$,
$t_{ij}$, $u_{ij}$.  It turns out that the BTEs are exactly satisfied
for Teichm\"uller TQFT:

\begin{proposition}
  $e^{i\Delta_\sigma} = 1$.
\end{proposition}

\begin{proof}
  Let us embed each of the shaped rhombic dodecahedra representing the
  two sides of BTE[$\sigma$] into a larger shaped pseudo $3$-manifold
  so that all edges become internal.  We do this in such a way that
  the obtained shaped pseudo $3$-manifolds differ only in the shaped
  triangulations of the embedded rhombic dodecahedra; therefore, their
  partition functions differ by the same phase factor
  $e^{i\Delta_\sigma}$.

  Recall that the parameters $s_i$, $t_{ij}$, $u_{ij}$ were introduced
  by shape gauge transformations.  Thus, we can shift any of them by a
  shape gauge transformation, and if we do so, the partition functions
  of the above two shaped pseudo $3$-manifolds are multiplied by phase
  factors according to \eqref{eq:SGI-Teichmuller}.  It is easy to check
  that these phase factors coincide.  Therefore, $e^{i\Delta_\sigma}$
  is independent of the parameters.

  Let us consider the case in which $u_{ij} = u_{kl}$ for
  $\{i,j\} \sqcup \{k, l\} = \{1,2,3,4\}$.

  For BTE[1], the two shaped rhombic dodecahedra have identical shaped
  triangulations in this case. Hence, $e^{i\Delta_1} = 1$.

  For BTE[0], the bipyramid with vertices
  $\{a_i, a_j, b_k, b_l, c_m, d\}$ on the left-hand side and the
  bipyramid with vertices $\{b_j, b_i, a_l, a_k, d, c_m\}$ on the
  right-hand side have exactly the same shaped triangulation.
  Therefore, during the sequence of $2$--$3$ and $3$--$2$ moves that
  transforms the left-hand side to the right-hand side, the phase
  factor produced by the $2$--$3$ move on the former is canceled by
  the phase factor produced by the $3$--$2$ move to obtain the latter.
\end{proof}

\section{Conclusions}
\label{sec:conclusions}

In this work, we presented a construction of three-dimensional
bicolored lattice models using line defects in state integral models
on shaped pseudo $3$-manifolds.  The R-matrices of these models solve
the BTEs.

We have seen, however, that there are some obstacles to establishing
the integrability of these models.  Most notably, the parameters of
these R-matrices are introduced by shape gauge transformations, and
consequently the layer transfer matrix with simple periodic boundary
conditions lacks spectral parameters.  Twisting the boundary
conditions with symmetries of the R-matrices can introduce parameter
dependence to the transfer matrix (as done in \cite{Inoue:2025xzj}),
but how exactly this is done depends on the specific state integral
model used.

Even if spectral parameters could be introduced into the transfer
matrices, the R-matrices and the transfer matrices must still satisfy
additional conditions in order for the lattice models to be
integrable, as explained in section~\ref{sec:BLM}.

That said, we believe that the lattice models constructed in this work
are interesting whether or not they are integrable.  One reason is
that their R-matrices satisfy the BTEs, which we expect to possess
rich algebraic structures generalizing the quantum algebras in the
case of the Yang--Baxter equation.  Another reason is that
Teichm\"uller TQFT is thought to capture aspects of three-dimensional
quantum gravity, so the associated lattice model should also admit a
gravity interpretation.  We leave the exploration of these directions
for future work.

\acknowledgments{The authors are grateful to Hyun Kyu Kim and Nicolai Reshetikhin for
valuable discussions. This work is supported by the National Natural
Science Foundation of China (NSFC) under Grant No.~12375064. MS is
also supported in part by the Beijing Natural Science Foundation
(BJNSF) under Grant No.~IS24010 and by the Shuimu Scholar Program of
Tsinghua University. XS is supported by the Beijing Natural Science
Foundation (BJNSF) under Grant No.~IS25024.}

\bibliographystyle{JHEP}
\bibliography{Bibliography}

\end{document}